\documentclass[12pt,onecolumn]{IEEEtran}
\usepackage{graphicx,amssymb,amsmath}
\usepackage[noadjust]{cite}
\usepackage{setspace}               
\usepackage{stfloats}
\usepackage{bbding}
\usepackage{amsthm}
\usepackage{amsmath}
\usepackage{flushend}
\usepackage{cases,subeqnarray}
\usepackage{bm,multirow,bigstrut}
\usepackage{textcomp}
\usepackage{latexsym,bm}
\usepackage{booktabs,changebar}
\usepackage{xcolor}
\usepackage{mathtools}
\usepackage{dsfont}
\usepackage{extarrows}
\usepackage{mathrsfs}
\usepackage{subfigure}

\usepackage{cite}
\usepackage{bm}
\usepackage{cleveref}
\usepackage{multicol}       
\usepackage{multirow}       
\usepackage{array}          
\usepackage{colortbl}
\usepackage{makecell}
\definecolor{crimson}{RGB}{192,0,0}         
\definecolor{navy}{RGB}{47,85,151}         

\makeatletter                               
\newif\if@restonecol
\makeatother

\makeatletter
\newif\if@restonecol
\makeatother

\usepackage[linesnumbered,ruled,vlined]{algorithm2e}
\usepackage{algpseudocode}
\usepackage{amsmath}
\usepackage{listings}                       
\usepackage[framed,numbered,autolinebreaks,useliterate]{mcode}
\theoremstyle{plain}
\newtheorem{thm}{Theorem}
\newtheorem{lemm}{Lemma}

\newtheorem{coro}{Corollary}

\theoremstyle{plain}
\newtheorem{rem}{Remark}

\IEEEoverridecommandlockouts

\begin{document}


\title{Uplink Precoding Design for Cell-Free Massive MIMO with Iteratively Weighted MMSE
\thanks{This article was presented in part at IEEE International Conference on Communications 2022 \cite{2022arXiv220111299W}.}
\thanks{Z. Wang and J. Zhang are with the School of Electronic and Information Engineering, Beijing Jiaotong University, Beijing 100044, China, and also with the Frontiers Science Center for Smart High-speed Railway System, Beijing Jiaotong University, Beijing 100044, China (e-mail: \{zhewang\_77, jiayizhang\}@bjtu.edu.cn).}
\thanks{H. Q. Ngo is with the Institute of Electronics, Communications, and Information Technology, Queen’s University Belfast, BT3 9DT Belfast, U.K. (email: hien.ngo@qub.ac.uk).}
\thanks{B. Ai is with the State Key Laboratory of Rail Traffic Control and Safety, Beijing Jiaotong University, Beijing 100044, China, also with the Frontiers Science Center for Smart High-Speed Railway System and the Henan Joint International Research Laboratory of Intelligent Networking and Data Analysis, Zhengzhou University, Zhengzhou 450001, China, and also with the Research Center of Networks and Communications, Peng Cheng Laboratory, Shenzhen 518066, China (e-mail: boai@bjtu.edu.cn).}
\thanks{M. Debbah is with the Technology Innovation Institute, Abu Dhabi, United Arab Emirates, and also with CentraleSup{\'e}lec,
University Paris-Saclay, 91192 Gif-sur-Yvette, France (e-mail: merouane.debbah@tii.ae).}}
\author{Zhe Wang, Jiayi Zhang,~\IEEEmembership{Senior Member,~IEEE,} Hien Quoc Ngo,~\IEEEmembership{Senior Member,~IEEE,} Bo Ai,~\IEEEmembership{Fellow,~IEEE} and M{\'e}rouane Debbah,~\IEEEmembership{Fellow,~IEEE}}
\maketitle
\vspace{-1.5cm}

\maketitle


\begin{abstract}

In this paper, we investigate a cell-free massive multiple-input multiple-output system with both access points and user equipments equipped with multiple antennas over the Weichselberger Rayleigh fading channel. We study the uplink spectral efficiency (SE) for the fully centralized processing scheme and large-scale fading decoding (LSFD) scheme. To further improve the SE performance, we design the uplink precoding schemes based on the weighted sum SE maximization. Since the weighted sum SE maximization problem is not jointly over all optimization variables, two efficient uplink precoding schemes based on Iteratively Weighted sum-Minimum Mean Square Error (I-WMMSE) algorithms, which rely on the iterative minimization of weighted MSE, are proposed for two processing schemes investigated. Furthermore, with maximum ratio combining applied in the LSFD scheme, we derive novel closed-form achievable SE expressions and optimal precoding schemes. Numerical results validate the proposed results and show that the I-WMMSE precoding schemes can achieve excellent sum SE performance with a large number of UE antennas.
\end{abstract}
\begin{IEEEkeywords}
Cell-free massive MIMO, uplink precoding, weighted sum-rate maximization, spectral efficiency.
\end{IEEEkeywords}

\IEEEpeerreviewmaketitle
\vspace{-0.6cm}
\section{Introduction}
\vspace{-0.3cm}
Cell-free massive multiple-input multiple-output (CF mMIMO) has attracted a lot of research interest and is regarded as a promising technology for future wireless communications, for its ability to achieve uniformly high spectral efficiency (SE) \cite{7827017,9113273,chen2021survey,[162],9586055,8768014}. Basically, a large number of access points (APs), arbitrarily distributed in a wide coverage area and connected to one or several central processing units (CPUs), jointly serve all user equipments (UEs) on the same time-frequency resource. Compared with the traditional cellular mMIMO system, the CF mMIMO system operates with no cell boundaries and many more APs than UEs \cite{8097026,9064545,9174860}. Relying upon the prominent network topology of CF mMIMO, four uplink (UL) signal processing schemes, distinguished from levels of the mutual cooperation between all APs and the assistance from the CPU, can be implemented as \cite{[162]}. Among these schemes, the ``\emph{Level 4}"  and ``\emph{Level 3}" are viewed as efficient processing techniques. The so-called \emph{Level 4} is a fully-centralized processing scheme where all the pilot and data signals received at APs are transmitted to the CPU via the fronthaul links and the CPU performs channel estimation and data detection. The similar scheme was also investigated in \cite{9528977,9849114,9043895}. The so-called \emph{Level 3} stands for a two layer decoding scheme: in the first layer, each AP estimates channels and decodes the UE data locally by applying an arbitrary combining scheme based on the local channel state information (CSI); in the second layer, all the local estimates of the UE data are gathered at the CPU in which they are linearly weighted by the optimal large-scale fading decoding (LSFD) coefficient to obtain the final decoding data. The LSFD scheme has been widely investigated in \cite{7869024,8809413,9276421,9529197} since it can make full use of the prominent network topology for CF mMIMO and achieve excellent performance.

To promote the practical implementation of the CF mMIMO network, a new framework of scalable CF mMIMO system and its respective processing algorithms were proposed in \cite{9064545} by exploiting the dynamic cooperation cluster (DCC) concept. Besides, the scalability aspects in a realistic scenario with multiple CPUs were considered in \cite{8761828}, where the data processing, network topology and power control strategies with multiple CPUs were discussed. Moreover, the authors of \cite{9499049} considered the uplink of a radio-strip-based CF mMIMO network architecture with sequential fronthaul links between APs and proposed MMSE-based sequential processing schemes, which significantly reduced the fronthaul requirement. However, when the CF mMIMO network is operated in practice, a more practical capacity-constrained fronthaul network would have a great effect on the system performance. The authors of \cite{8891922} and \cite{9446982} discussed the uplink performance of a CF mMIMO system with limited capacity fronthaul links. Furthermore, it is worth noting that the CF mMIMO architecture has been co-designed with another promising future wireless technology: Reconfigurable Intelligent Surface (RIS) \cite{9665300,9743355}, which would undoubtedly provide vital tutorials for the future wireless network design.

The vast majority of scientific papers on CF mMIMO focus on the scenario with single-antenna UEs. However, in practice, contemporary UEs with moderate physical sizes have already been equipped with multiple antennas to achieve higher multiplexing gain and boost the system reliability. The authors of \cite{8646330} investigated the UL performance of a CF mMIMO system with multi-antenna UEs over maximum ratio (MR) combining and zero-forcing (ZF) combining. The authors of \cite{8901451} considered a user-centric (UC) approach for CF mMIMO with multi-antenna UEs and proposed power allocation strategies for either sum-rate maximization or minimum-rate maximization. Besides, the authors of \cite{194} analyzed the downlink SE performance for a CF mMIMO system with multi-antenna UEs and computed SE expressions in closed-form. Then, the SE performance for a CF mMIMO system with multi-antenna UEs and low-resolution DACs was investigated in \cite{9424703}. Nevertheless, these works only investigated a simple distributed processing scheme and are based on the overly idealistic assumption of independent and identically distributed (i.i.d.) Rayleigh fading channels, neglecting the spatial correlation that has a significant impact on practical CF mMIMO systems \cite{8809413,9276421}. The authors of \cite{04962} considered a CF mMIMO system with multi-antenna UEs over the jointly-correlated Weichselberger model \cite{1576533} and analyzed four UL processing schemes.

As observed in \cite{194,04962}, increasing the number of antennas per UE may not always benefit the SE performance. The SE would reach the maximum value with particular number of antennas per UE, then decrease with the increase of number of antennas per UE. One main reason for this phenomenon is that the UEs cannot make full use of the benefit of equipping with multiple antennas to achieve higher SE performance without UL precoding schemes. So it is undoubtedly vital to design the UL precoding scheme to further improve the performance of systems. However, it is worth noting that the design of UL precoding for CF mMIMO has not been investigated. For the traditional mMIMO or MIMO systems, one popular optimization objective for the uplink/downlink precoding design is to maximize the weighted sum rate (WSR) \cite{5756489,6302110,4712693,9347738}. The authors of \cite{5756489} and \cite{4712693} discussed the equivalence between the WSR maximization problem and the Weighted sum-Minimum Mean Square Error (WMMSE) problem in MIMO systems and proposed an iteratively downlink transceiver design algorithm for the WSR maximization. Note that the algorithm relies on the iterative minimization of weighted MSE since the WMMSE problem are not jointly convex over all optimization variables. Moreover, the authors of \cite{6302110} investigated the UL precoding scheme optimization based on \cite{5756489} under sum-power-constraint or individual-power-constraint.

Motivated by the above observations, we investigate a CF mMIMO system with both multi-antenna APs and UEs over the Weichselberger Rayleigh fading  channel. Two pragmatic processing schemes: 1) the fully centralized processing scheme; 2) the large-scale fading decoding scheme are implemented. The main contributions are given as follows.
\begin{itemize}
\item We design an efficient UL precoding scheme to maximize the WSR for the fully centralized processing scheme based on an iteratively WMMSE (I-WMMSE) algorithm. Note that the design of I-WMMSE precoding scheme for the fully centralized processing scheme is implemented at the CPU and based on the instantaneous CSI.
\item For the LSFD processing scheme, we derive a UL precoding scheme for the WSR maximization based on an iteratively WMMSE algorithm. The design of I-WMMSE precoding scheme for the LSFD scheme is implemented at the CPU but based only on channel statistics. More importantly, we compute achievable SE expressions and optimal precoding schemes in novel closed-form for the LSFD scheme with MR combining.
\item We analyze the practical implementation and computation complexity for the proposed I-WMMSE precoding schemes. It is found that the proposed I-WMMSE precoding schemes can be guaranteed to converge. More importantly, the proposed UL precoding schemes are efficient to achieve excellent sum SE/rate performance and the average rate benefits from the multiple antennas at the UE-side, which undoubtedly provides vital insights for the practical implementation of multi-antenna UEs.
\end{itemize}

Note that this paper differs from the conference version \cite{2022arXiv220111299W} in the following aspects: \emph{i}) we investigate the fully centralized processing/LSFD schemes and design their respective UL precoding schemes, while only the LSFD scheme was considered in \cite{2022arXiv220111299W}; \emph{ii}) we provide details for the derivation of the I-WMMSE precoding schemes, which are omitted in \cite{2022arXiv220111299W} due to the lack of space; \emph{iii}) we analyze the practical implementation and convergence behavior of the proposed precoding schemes. More importantly, numerical results show vital insights for the CF mMIMO system with the proposed UL precoding schemes.

The rest of this paper is organized as follows. In Section \ref{sec:system}, we consider a CF mMIMO system with the Weichselberger Rayleigh fading channel, and describe the channel estimation and data detection. Then, Section \ref{sec:Iterative Optimization} introduces the fully centralized processing and LSFD processing schemes, and provides their respective achievable SE expressions. Novel closed-form SE expressions for the LSFD scheme with MR combining are derived. More importantly, based on the achievable SE expressions, we propose UL I-WMMSE precoding schemes for two processing schemes. Then, Section \ref{Precoding Implementation} provides some insights for the practical implementation and computation complexity of proposed I-WMMSE precoding schemes. In Section \ref{Numerical_Results}, numerical results and performance analysis for the I-WMMSE precoding schemes are provided. Finally, the major conclusions and future directions are drawn in Section \ref{sec:conclusion}.

\textbf{\emph{Notation}}: Lowercase letters $\mathbf{x}$ and boldface uppercase letters $\mathbf{X}$ denote the column vectors and matrices, respectively. $\mathbb{E} \left\{ \cdot \right\}$, $\mathrm{tr}\left\{ \cdot \right\}$ and $\triangleq$ are the expectation operator, the trace operator, and the definitions, respectively. $\left| \cdot \right|$, $\left\| \cdot \right\|$ and $\left\| \cdot \right\| _{\mathrm{F}}$ are the determinant of a matrix or the absolute value of a number, the Euclidean norm and the Frobenious norm, respectively. $\mathrm{vec}\left( \mathbf{A} \right)$ denotes a column vector formed by the stack of the columns of $\mathbf{A}$. The $n\times n$ identity matrix is represented by $\mathbf{I}_{n\times n}$. The Kronecker products and the element-wise products are denoted by $\otimes$ and $\odot$, respectively. Finally, $\mathbf{x}\sim \mathcal{N} _{\mathbb{C}}\left( 0,\mathbf{R} \right)$ is a circularly symmetric complex Gaussian distribution vector with correlation matrix $\mathbf{R}$.

\vspace{-0.6cm}
\section{System Model}\label{sec:system}
\vspace{-0.3cm}
In this paper, we investigate a CF mMIMO system consisting of $M$ APs and $K$ UEs, where all APs are connected to one or several CPUs via fronthaul links as shown in Fig. \ref{System_Model}. For simplicity, there is only one CPU and all APs serve all UEs\footnote{As shown in Fig. \ref{System_Model}, a more practical network topology is with multiple CPUs and dynamic cooperation clusters, where each UE is only served by a cluster of APs and the APs are grouped into cell-centric clusters. Each cell-centric cluster is connected to a particular CPU.}. The numbers of antennas per AP and UE are $L$ and $N$, respectively. A standard block fading model is investigated, in which the channel response is constant and frequency flat in a coherence block of $\tau _c$-length (channel uses). Let $\tau _p$ and $\tau _c-\tau _p$ denote channel uses dedicated for the channel estimation and data transmission, respectively. We denote by $\mathbf{H}_{mk}\in \mathbb{C}^{L\times N}$ the channel response between AP $m$ and UE $k$. We assume that $\mathbf{H}_{mk}$ for different AP-UE pairs are independent.
\begin{figure}[t]
\centering
\includegraphics[scale=0.6]{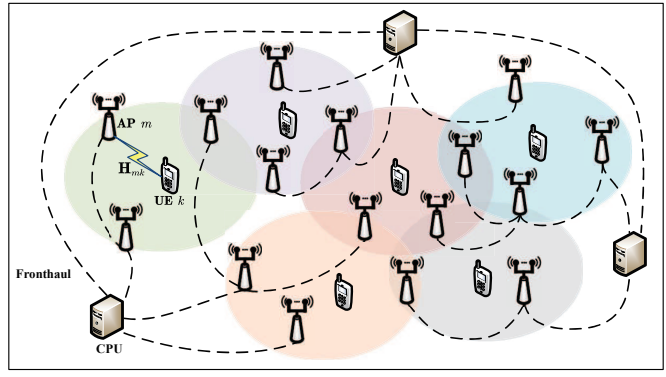}
\caption{A cell-free massive MIMO system.
\label{System_Model}}
\vspace{-0.6cm}
\end{figure}
\vspace{-0.6cm}
\subsection{Channel Model}
Based on the jointly-correlated (also known as the Weichselberger model \cite{1576533}) Rayleigh fading channel\footnote{Note that the Rayleigh fading channel is a special case of the Rician fading channel. And the performance gap between the Rician channel and the Rayleigh channel is small \cite{8645336}. However, the focus of this paper is not on the channel model but on the UL precoding scheme design. So for the simplicity of analysis, we investigate an essential Rayleigh fading channel by assuming there is no line-of-sight (LoS) link between each UE and AP.}, $\mathbf{H}_{mk}$ is modeled as
\begin{equation}\label{Weichselberger}
\mathbf{H}_{mk}=\mathbf{U}_{mk,\mathrm{r}}\left( \mathbf{\tilde{\Omega}}_{mk}\odot \mathbf{H}_{mk,\mathrm{iid}} \right) \mathbf{U}_{mk,\mathrm{t}}^{H}
\end{equation}
where $\mathbf{U}_{mk,\mathrm{r}}=\left[ \mathbf{u}_{mk,\mathrm{r},1},\cdots ,\mathbf{u}_{mk,\mathrm{r},L} \right] \in \mathbb{C}^{L\times L}$ and $\mathbf{U}_{mk,\mathrm{t}}=\left[ \mathbf{u}_{mk,\mathrm{t},1},\cdots ,\mathbf{u}_{mk,\mathrm{t},N} \right] \in \mathbb{C}^{N\times N}$ are the eigenvector matrices of the one-sided correlation matrices $\mathbf{R}_{mk,\mathrm{r}}\triangleq \mathbb{E}\left[ \mathbf{H}_{mk}\mathbf{H}_{mk}^{H} \right]$ and $\mathbf{R}_{mk,\mathrm{t}}\triangleq \mathbb{E}\left[ \mathbf{H}_{mk}^{T}\mathbf{H}_{mk}^{*} \right]$, and $\mathbf{H}_{mk,\mathrm{iid}}\in \mathbb{C}^{L\times N}$ is composed of i.i.d. $\mathcal{N}_{\mathbb{C}}\left( 0,1 \right)$ random entries, respectively. Besides, we denote by $\mathbf{\Omega }_{mk}\triangleq \mathbf{\tilde{\Omega}}_{mk}\odot \mathbf{\tilde{\Omega}}_{mk}\in \mathbb{R} ^{L\times N}$ the ``eigenmode coupling matrix'' with the $\left( l,n \right)$-th element $\left[ \mathbf{\Omega }_{mk} \right] _{ln}$ specifying the average amount of power coupling from $\mathbf{u}_{mk,\mathrm{r},l}$ to $\mathbf{u}_{mk,\mathrm{t},n}$. $\mathbf{H}_{mk}$ can also be formed as $\mathbf{H}_{mk}=\left[ \mathbf{h}_{mk,1},\cdots ,\mathbf{h}_{mk,N} \right] $ with $\mathbf{h}_{mk,n}\in \mathbb{C} ^L$ being the channel between AP $m$ and $n$-th antenna of UE $k$. By stacking the columns of $\mathbf{H}_{mk}$ on each other, we define $\mathbf{h}_{mk}\triangleq \mathrm{vec}\left( \mathbf{H}_{mk} \right) = [ \mathbf{h}_{mk,1}^{T},\cdots ,\mathbf{h}_{mk,N}^{T}] ^T\sim \mathcal{N} _{\mathbb{C}}\left( 0,\mathbf{R}_{mk} \right)$, where $\mathbf{R}_{mk}\triangleq \mathbb{E} \{ \mathbf{h}_{mk}\mathbf{h}_{mk}^{H} \} $ is the full correlation matrix
\begin{equation}
\mathbf{R}_{mk}=( \mathbf{U}_{mk,\mathrm{t}}^{*}\otimes \mathbf{U}_{mk,\mathrm{r}} ) \mathrm{diag}\left( \mathrm{vec}\left( \mathbf{\Omega }_{mk} \right) \right) ( \mathbf{U}_{mk,\mathrm{t}}^{*}\otimes \mathbf{U}_{mk,\mathrm{r}} ) ^H.
\end{equation}
Moreover, note that $\mathbf{R}_{mk}$ can be structured into the block form as \cite{04962} with the $(n,i)$-th submatrix being $\mathbf{R}_{mk}^{ni}=\mathbb{E} \{ \mathbf{h}_{mk,n}\mathbf{h}_{mk,i}^{H}\}$. Besides, the large-scale fading coefficient $\beta _{mk}$ can be extracted from $\mathbf{R}_{mk}$ as $\beta _{mk}=\frac{1}{LN}\mathrm{tr}\left( \mathbf{R}_{mk} \right) =\frac{1}{LN}\left\| \mathbf{\Omega }_{mk} \right\| _1$. It is worth mentioning that the motivations for adopting the Weichselberger channel model are: 1) The Weichselberger model investigated in \eqref{Weichselberger} not only captures the correlation features at both the AP-side and UE-side but models the joint correlation dependence between each AP-UE pair through the coupling matrix; 2) The coupling matrix $\mathbf{\Omega}_{mk}$ reflects the practical spatial arrangement of scattering objects between AP $m$ and UE $k$. More significantly, the Weichselberger model can reduce to most channel models of great interest by adjusting the coupling $\mathbf{\Omega}_{mk}$ to particular formulation, such as the Kronecker model and i.i.d. Rayleigh fading model \cite{1576533,04962}; 3) Compared with other stochastic channel models, the Weichselberger model displays significantly less modeling error, which is validated based on the practical measurement in \cite{1576533}.

\vspace*{-0.7cm}
\subsection{Channel Estimation}
\vspace*{-0.3cm}
For the channel estimation, mutually orthogonal pilot matrices are constructed and each pilot matrix is composed of $N$ mutually orthogonal pilot sequences.
We denote by $\mathbf{\Phi }_k$ the pilot matrix assigned to UE $k$ with $\boldsymbol{\Phi }_{k}^{H}\boldsymbol{\Phi }_l=\tau _p\mathbf{I}_N$, if $\ l=k$ and ${\bf{0}}$ otherwise.
And $\mathcal{P}_k$ is the index subset of UEs using the same pilot matrix as UE $k$ including itself. When all UEs transmit their pilot matrices, the received signal at AP $m$ $\mathbf{Y}_{mk}^{\mathrm{p}}\in \mathbb{C} ^{L\times \tau _p}$ is
$\mathbf{Y}_{m}^{\mathrm{p}}=\sum_{k=1}^K{\mathbf{H}_{mk}\mathbf{F}_{k,\mathrm{p}}\mathbf{\Phi }_{k}^{T}+\mathbf{N}_{m}^{\mathrm{p}}},$
where $\mathbf{F}_{k,\mathrm{p}}\in \mathbb{C} ^{N\times N}$ is the precoding matrix for UE $k$ under the phase of pilot transmission, $\mathbf{N}_{m}^{\mathrm{p}}\in \mathbb{C}^{L\times \tau _p}$ is the additive noise at AP $m$ with independent $\mathcal{N}_{\mathbb{C}}( 0,\sigma ^2 )$ entries and $\sigma ^2$ being the noise power, respectively. The pilot transmission should be implemented under the power constraint as $\mathrm{tr}( \mathbf{F}_{k,\mathrm{p}}\mathbf{F}_{k,\mathrm{p}}^{H}) \leqslant p_k$, where $p_k$ is the maximum transmit power for UE $k$. To derive sufficient statistics for $\mathbf{h}_{mk}$, AP $m$ projects $\mathbf{Y}_{mk}^{\mathrm{p}}$ onto $\mathbf{\Phi }_{k}^{*}$ as
$
\mathbf{Y}_{mk}^{p}=\mathbf{Y}_{m}^{p}\boldsymbol{\Phi }_{k}^{*}=\sum_{l=1}^K{\mathbf{H}_{ml}\mathbf{F}_{l,\mathrm{p}}\left( \boldsymbol{\Phi }_{l}^{T}\boldsymbol{\Phi }_{k}^{*} \right) +}\mathbf{N}_{m}^{\mathrm{p}}\boldsymbol{\Phi }_{k}^{*}=\sum_{l\in \mathcal{P}_k}{\tau _p\mathbf{H}_{ml}\mathbf{F}_{l,\mathrm{p}}}+\mathbf{Q}_{mk}^{\mathrm{p}},
$
where $\mathbf{Q}_{mk}^{\mathrm{p}}\triangleq \mathbf{N}_{m}^{p}\boldsymbol{\Phi }_{k}^{*}$. Then, following the standard MMSE estimation steps in \cite{5340650} and \cite{8187178}, AP $m$ can compute the MMSE estimation of $\mathbf{h}_{mk}$ as
\begin{equation}\label{eq:Local_Estimation}
\mathbf{\hat{h}}_{mk}=\mathrm{vec}( \mathbf{\hat{H}}_{mk} ) =\mathbf{R}_{mk}\mathbf{\tilde{F}}_{k,\mathrm{p}}^{H}\mathbf{\Psi }_{mk}^{-1}\mathbf{y}_{mk}^{\mathrm{p}},
\end{equation}
where $\mathbf{\hat{H}}_{mk}$ is the MMSE estimation of $\mathbf{H}_{mk}$, $\mathbf{\tilde{F}}_{k,\mathrm{p}}=\mathbf{F}_{k,\mathrm{p}}^{T}\otimes \mathbf{I}_L$, $\mathbf{y}_{mk}^{\mathrm{p}}\triangleq\mathrm{vec}\left( \mathbf{Y}_{mk}^{\mathrm{p}} \right)=\sum_{l\in \mathcal{P} _k}^{}{\tau _p\mathbf{\tilde{F}}_{l,\mathrm{p}}\mathbf{h}_{ml}+\mathbf{q}_{m}^{\mathrm{p}}}$, $\mathbf{q}_{m}^{\mathrm{p}}=\mathrm{vec}\left(\mathbf{Q}_{mk}^{\mathrm{p}}\right)$ and $\mathbf{\Psi }_{mk}=\sum\nolimits_{l\in \mathcal{P} _k}^{}{\tau _p\mathbf{\tilde{F}}_{l,\mathrm{p}}\mathbf{R}_{ml}\mathbf{\tilde{F}}_{l,\mathrm{p}}^{H}}+\sigma ^2\mathbf{I}_{LN}$, respectively. Note that the estimate $\mathbf{\hat{h}}_{mk}$ and estimation error $\mathbf{\tilde{h}}_{mk}=\mathbf{h}_{mk}-\mathbf{\hat{h}}_{mk}$ are independent random vectors distributed as $\mathbf{\hat{h}}_{mk}\sim \mathcal{N} _{\mathbb{C}}( \mathbf{0},\mathbf{\hat{R}}_{mk} )$ and $\mathbf{\tilde{h}}_{mk}\sim \mathcal{N} _{\mathbb{C}}( \mathbf{0},\mathbf{C}_{mk} )$, where $\mathbf{\hat{R}}_{mk}\triangleq \tau _p\mathbf{R}_{mk}\mathbf{\tilde{F}}_{k,\mathrm{p}}^{H}\mathbf{\Psi}_{mk}^{-1}\mathbf{\tilde{F}}_{k,\mathrm{p}}\mathbf{R}_{mk}$ and $\mathbf{C}_{mk}\triangleq \mathbf{R}_{mk}-\mathbf{\hat{R}}_{mk}$. We can also form $\mathbf{\hat{R}}_{mk}$ and $\mathbf{C}_{mk}$ in the block structure with the $(n,i)$-th submatrix being $\mathbf{\hat{R}}_{mk}^{ni}=\mathbb{E} \{ \mathbf{\hat{h}}_{mk,n}\mathbf{\hat{h}}_{mk,i}^{H} \} $ and $\mathbf{C}_{mk}^{ni}=\mathbb{E} \{ \mathbf{\tilde{h}}_{mk,n}\mathbf{\tilde{h}}_{mk,i}^{H} \}$, respectively.
\vspace*{-0.7cm}
\subsection{Data Transmission}
\vspace*{-0.3cm}
For the data transmission, all antennas of all UEs simultaneously transmit their data symbols to all APs. The received signal $\mathbf{y}_m\in \mathbb{C} ^L$ at AP $m$ is
\begin{equation}
\mathbf{y}_m=\sum_{k=1}^K{\mathbf{H}_{mk}\mathbf{s}_k}+\mathbf{n}_m,
\end{equation}
where $\mathbf{n}_m\sim \mathcal{N} _{\mathbb{C}}( 0,\sigma ^2\mathbf{I}_L ) $ is the independent receiver noise. The transmitted signal from UE $k$ $\mathbf{s}_k\in \mathbb{C} ^N$ can be constructed as $\mathbf{s}_k=\mathbf{F}_{k,\mathrm{u}}\mathbf{x}_k$, where $\mathbf{x}_k\sim \mathcal{N} _{\mathbb{C}}( 0,\mathbf{I}_N )$ is the data symbol for UE $k$ and $\mathbf{F}_{k,\mathrm{u}}\in \mathbb{C} ^{N\times N}$ is the precoding matrix for the data transmission which should satisfy the power constraint of UE $k$ as $\mathrm{tr}( \mathbf{F}_{k,\mathrm{u}}\mathbf{F}_{k,\mathrm{u}}^{H} ) \leqslant p_k$.
\vspace*{-0.5cm}
\section{Spectral Efficiency Analysis and I-WMMSE Precoding Design}\label{sec:Iterative Optimization}
\vspace*{-0.3cm}
In this section, we investigate two promising signal processing schemes, called ``fully centralized processing" and ``LSFD processing", and analyze their corresponding SE performance and design respective iteratively WMMSE precoding schemes\footnote{We only optimize the precoding matrices for the phase of data transmission $\mathbf{F}_{k,\mathrm{u}}$. The optimization of  $\mathbf{F}_{k,\mathrm{p}}$ is left for future research. Although we do not design $\mathbf{F}_{k,\mathrm{p}}$ in this paper, we try to keep the derived equations more generalized. So a scenario with arbitrary $\mathbf{F}_{k,\mathrm{p}}$ instead of limiting $\mathbf{F}_{k,\mathrm{p}}$ to a particular form is investigated. It is worth noting that all equations in this paper hold for any $\mathbf{F}_{k,\mathrm{p}}$ so undoubtedly provide some important guidelines for the investigation of optimization design for $\mathbf{F}_{k,\mathrm{p}}$ in the future work.}.
\vspace*{-0.7cm}
\subsection{Fully Centralized Processing}
\vspace*{-0.3cm}
\subsubsection{Spectral Efficiency Analysis}
For the fully centralized processing scheme, all $M$ APs send all the received pilot signals and data signals to the CPU. Indeed, both the channel estimation and data detection are implemented at the CPU. The collective channel $\mathbf{h}_k\in \mathbb{C}^{MLN}$ for UE $k$ can be constructed as $\mathbf{h}_k=[ \mathrm{vec}( \mathbf{H}_{1k})^T ,\cdots ,\mathrm{vec}( \mathbf{H}_{Mk} )^T ] ^T\sim \mathcal{N} _{\mathbb{C}}( \mathbf{0},\mathbf{R}_k)$ with $\mathbf{R}_k=\mathrm{diag}\left( \mathbf{R}_{1k},\cdots ,\mathbf{R}_{Mk} \right) \in \mathbb{C}^{MLN\times MLN}$ being the whole block-diagonal correlation matrix for UE $k$. Similar to \eqref{eq:Local_Estimation}, the CPU can derive the channel estimate for UE $k$ as\footnote{Note that the pilot signals received at the APs are first transmitted to the CPU and then the CPU estimates the channels, where $\tau _pML$ complex scalars are sent from the APs to the CPU at each coherence block. Alternatively, all APs can first estimate the channels as \eqref{eq:Local_Estimation}, and then send their channel estimates to the CPU, where $MKLN$ complex scalars are sent from the APs to the CPU at each coherence block. Since the pilot contamination is investigated ($\tau _p<KN$) in this paper, we consider the first transmission protocol due to its lower fronthaul overhead.}
$
\mathbf{\hat{h}}_k\triangleq \left[ \mathbf{\hat{h}}_{1k}^{T},\dots ,\mathbf{\hat{h}}_{Mk}^{T} \right] ^T\sim \mathcal{N} _{\mathbb{C}}\left( \mathbf{0},\tau _p\mathbf{R}_k\mathbf{\bar{F}}_{k,\mathrm{p}}^{H}\mathbf{\Psi }_{k}^{-1}\mathbf{\bar{F}}_{k,\mathrm{p}}\mathbf{R}_k \right)
$
where $\mathbf{\bar{F}}_{k,\mathrm{p}}=\mathrm{diag}( \underset{M}{\underbrace{\mathbf{\tilde{F}}_{k,\mathrm{p}},\dots ,\mathbf{\tilde{F}}_{k,\mathrm{p}}}})$ and $\mathbf{\Psi }_{k}^{-1}=\mathrm{diag}( \mathbf{\Psi }_{1k}^{-1},\dots ,\mathbf{\Psi }_{Mk}^{-1})$. The channel estimation error is $\mathbf{\tilde{h}}_k\sim \mathcal{N} _{\mathbb{C}}\left( 0,\mathbf{C}_k \right) $ where $\mathbf{C}_k\triangleq \mathbf{R}_k-\tau _p\mathbf{R}_k\mathbf{\bar{F}}_{k,\mathrm{p}}^{H}\mathbf{\Psi }_{k}^{-1}\mathbf{\bar{F}}_{k,\mathrm{p}}\mathbf{R}_k$. Moreover, the received data signal at the CPU can be denoted as
\begin{equation}
\underset{=\,\,\mathbf{y}}{\underbrace{[ \mathbf{y}_{1}^{T},\cdots ,\mathbf{y}_{M}^{T} ] ^T}}=\sum_{k=1}^K{\underset{=\,\,\mathbf{H}_k}{\underbrace{[ \mathbf{H}_{1k}^{T},\dots ,\mathbf{H}_{Mk}^{T} ] ^T}}\mathbf{F}_{k,\mathrm{u}}\mathbf{x}_k}+\underset{=\,\,\mathbf{n}}{\underbrace{[ \mathbf{n}_{1}^{T},\dots ,\mathbf{n}_{M}^{T} ] ^T}},
\end{equation}
or a compact form as
$
\mathbf{y}=\sum_{k=1}^K{\mathbf{H}_k\mathbf{F}_{k,\mathrm{u}}\mathbf{x}_k}+\mathbf{n}.
$

Under the setting of ``fully centralized processing", we assume that UL precoding matrices ($\mathbf{F}_{k,\mathrm{u}}$ and $\mathbf{F}_{k,\mathrm{p}}$) are available at the CPU. Based on the collective channel estimates, the CPU designs an arbitrary receive combining matrix $\mathbf{V}_k\in \mathbb{C} ^{LM\times N}$ for UE $k$ to detect $\mathbf{x}_k$ as
\begin{equation}
\mathbf{\check{x}}_k=\mathbf{V}_{k}^{H}\mathbf{y}=\mathbf{V}_{k}^{H}\mathbf{\hat{H}}_k\mathbf{F}_{k,\mathrm{u}}\mathbf{x}_k+\mathbf{V}_{k}^{H}\mathbf{\tilde{H}}_k\mathbf{F}_{k,\mathrm{u}}\mathbf{x}_k+\sum_{l\ne k}^K{\mathbf{V}_{k}^{H}\mathbf{H}_l\mathbf{F}_{l,\mathrm{u}}\mathbf{x}_l}+\mathbf{V}_{k}^{H}\mathbf{n},
\end{equation}
and the conditional MSE matrix for UE $k$ is
\begin{equation}\label{MSE_Matrix_1}
\begin{aligned}
&\mathbf{E}_{k,(1)}=\mathbb{E} \{(\mathbf{x}_k-\mathbf{\check{x}}_k)(\mathbf{x}_k-\mathbf{\check{x}}_k)^H|\{ \mathbf{\hat{H}}_k \} ,\{ \mathbf{F}_{k,\mathrm{u}} \} \}\\
&=\mathbf{I}_N-\mathbf{V}_{k}^{H}\mathbf{\hat{H}}_k\mathbf{F}_{k,\mathrm{u}}-\mathbf{F}_{k,\mathrm{u}}^{H}\mathbf{\hat{H}}_{k}^{H}\mathbf{V}_k+\mathbf{V}_{k}^{H}\left( \sum_{l=1}^K{\left( \mathbf{\hat{H}}_l\mathbf{\bar{F}}_{l,\mathrm{u}}\mathbf{\hat{H}}_{l}^{H}+\mathbf{C}_{l}^{\prime} \right)}+\sigma ^2\mathbf{I}_{ML} \right) \mathbf{V}_k
\end{aligned}
\end{equation}
where $\mathbf{\bar{F}}_{l,\mathrm{u}}\triangleq \mathbf{F}_{l,\mathrm{u}}\mathbf{F}_{l,\mathrm{u}}^{H}$, $\mathbf{C}_{l}^{\prime}\triangleq \mathrm{diag}\left( \mathbf{C}_{1l}^{\prime},\cdots ,\mathbf{C}_{Ml}^{\prime} \right) \in \mathbb{C}^{ML\times ML}$ and $\mathbf{C}_{ml}^{\prime}=\mathbb{E} \{ \mathbf{\tilde{H}}_{ml}\mathbf{\bar{F}}_{l,\mathrm{u}}\mathbf{\tilde{H}}_{ml}^{H} \} \in \mathbb{C} ^{L\times L}$ with the $\left( j,q \right) $-th element of $\mathbf{C}_{ml}^{\prime}$ being $\left[ \mathbf{C}_{ml}^{\prime} \right] _{jq}=\sum_{p_1=1}^N{\sum_{p_2=1}^N{\left[ \mathbf{\bar{F}}_l \right] _{p_2p_1}\left[ \mathbf{C}_{ml}^{p_2p_1} \right] _{jq}}}.$

By implementing the per-user-basis minimum mean-squared error-based successive interference cancellation (MMSE-SIC) detector while treating co-user interference as uncorrelated Gaussian noise, we derive the achievable SE for UE $k$ as follows.
\begin{coro}\label{Centralized_SE}
An achievable for UE $k$ under the setting of ``fully centralized processing" with the MMSE estimator is
\begin{equation}\label{eq:SE_4}
\mathrm{SE}_{k,(1)}=\left( 1-\frac{\tau _p}{\tau _c} \right) \mathbb{E}\left\{ \log _2\left| \mathbf{I}_N+\mathbf{D}_{k,\left( 1 \right)}^{H}\mathbf{\Sigma }_{k,\left( 1 \right)}^{-1}\mathbf{D}_{k,\left(1 \right)} \right| \right\} ,
\end{equation}
where $\mathbf{D}_{k,\left( 1 \right)}\triangleq \mathbf{V}_{k}^{H}\mathbf{\hat{H}}_k\mathbf{F}_{k,\mathrm{u}}$ and
$
\mathbf{\Sigma}_{k,(1)}\triangleq\mathbf{V}_{k}^{H}\left(\sum_{l=1}^K{\mathbf{\hat{H}}_l\mathbf{\bar{F}}_{l,\mathrm{u}}\mathbf{\hat{H}}_{l}^{H}}-\mathbf{\hat{H}}_k\mathbf{\bar{F}}_{k,\mathrm{u}}\mathbf{\hat{H}}_{k}^{H}+\sum_{l=1}^K{\mathbf{C}_{l}^{\prime}}+\sigma ^2\mathbf{I}_{ML} \right) \mathbf{V}_k.
$
The expectations are with respect to all sources of randomness.
\end{coro}
\begin{proof}
The proof follows from the similar approach as \cite[Corollary 1]{04962} and is therefore omitted.
\end{proof}

We notice that Corollary~\ref{Centralized_SE} holds for any combining schemes. One promising combining scheme is the MMSE combining as
\begin{equation}\label{eq:MMSE_Combining}
\mathbf{V}_k^{\mathrm{MMSE}}=\left( \sum_{l=1}^K{\left( \mathbf{\hat{H}}_l\mathbf{\bar{F}}_{l,\mathrm{u}}\mathbf{\hat{H}}_{l}^{H}+\mathbf{C}_{l}^{\prime} \right)}+\sigma ^2\mathbf{I}_{ML} \right) ^{-1}\mathbf{\hat{H}}_k\mathbf{F}_{k,\mathrm{u}},
\end{equation}which can minimize the mean-squared error $\mathrm{MSE}_{k,(1)}=\mathrm{tr}( \mathbf{E}_{k,(1)})$. With the MMSE combining scheme, the conditional MSE matrix in \eqref{MSE_Matrix_1} is
\begin{equation}\label{eq:MSE_Matrix_Optimal_1}
\mathbf{E}_{k,(1)}^{\mathrm{opt}}=\mathbf{I}_N-\mathbf{F}_{k,\mathrm{u}}^{H}\mathbf{\hat{H}}_{k}^{H}\left( \sum_{l=1}^K{\left( \mathbf{\hat{H}}_l\mathbf{\bar{F}}_{l,\mathrm{u}}\mathbf{\hat{H}}_{l}^{H}+\mathbf{C}_{l}^{\prime} \right)}+\sigma ^2\mathbf{I}_{ML} \right) ^{-1}\mathbf{\hat{H}}_k\mathbf{F}_{k,\mathrm{u}}
\end{equation}

More importantly, the MMSE combining in \eqref{eq:MMSE_Combining} can also maximize the achievable SE in \eqref{eq:SE_4} as follows.
\begin{coro}\label{MMSE_Optimal}
The achievable SE for UE $k$ in \eqref{eq:SE_4} can be maximized by the MMSE combining scheme in \eqref{eq:MMSE_Combining} with the maximum value
\begin{equation}
\mathrm{SE}_{k,(1)}^{\mathrm{opt}}=\left( 1-\frac{\tau _p}{\tau _c} \right) \mathbb{E} \left\{ \log _2\left| \mathbf{I}_N+\mathbf{F}_{k,\mathrm{u}}^{H}\mathbf{\hat{H}}_{k}^{H}\left( \sum_{l=1}^K{\left( \mathbf{\hat{H}}_l\mathbf{\bar{F}}_{l,\mathrm{u}}\mathbf{\hat{H}}_{l}^{H}+\mathbf{C}_{l}^{\prime} \right)}-\mathbf{\hat{H}}_k\mathbf{\bar{F}}_{k,\mathrm{u}}\mathbf{\hat{H}}_{k}^{H}+\sigma ^2\mathbf{I}_{ML} \right) ^{-1}\mathbf{\hat{H}}_k\mathbf{F}_{k,\mathrm{u}} \right| \right\}.
\end{equation}
\end{coro}
\begin{proof}
The proof can be found in \cite[Appendix B]{04962} and is therefore omitted.
\end{proof}

\subsubsection{Iteratively WMMSE Precoding Design}
In this part, we design the uplink precoding scheme for the ``fully centralized processing". One popular weighted sum-rate maximization problem is investigated as\footnote{The notation $\mathbf{F}$ is short for $\{ \mathbf{F}_{k,\mathrm{u}} \} _{k=1,\dots ,K}$, denoting all variables $\mathbf{F}_{k,\mathrm{u}}$  with $k=1,\dots ,K$. Similar definitions are applied for $\mathbf{V}$, $\mathbf{A}$, $\mathbf{W}$, $\mathbf{S}$ in the following. In this section, we denote by $\mathbf{F}_{k,\mathrm{u},(1)}$ and $\mathbf{F}_{k,\mathrm{u},(2)}$ the UL precoding matrix of UE $k$ for the fully centralized processing and LSFD scheme, respectively.}
\begin{equation}\label{Sum_SE_Fully}
\begin{aligned}
&\underset{\left\{ \mathbf{F} \right\}}{\max}\sum_{k=1}^K{\mu_{k,(1)}\mathrm{SE}_{k,(1)}}\\
&\mathrm{s}.\mathrm{t}. \left\| \mathbf{F}_{k,\mathrm{u},(1)} \right\| ^2\leqslant p_k\,\,\forall k=1,\dots ,K
\end{aligned}
\end{equation}
where $\mu_{k,(1)}$ represents the priority weight of UE $k$ and $\mathrm{SE}_{k,(1)}$ is given by \eqref{eq:SE_4}.

As \cite{5756489} and \cite{4712693}, the matrix-weighted sum-MSE minimization problem as
\vspace*{-0.1cm}
\begin{equation}\label{Sum_MSE}
\vspace*{-0.1cm}
\begin{aligned}
&\underset{\left\{ \mathbf{F},\mathbf{V},\mathbf{W} \right\}}{\min}\sum_{k=1}^K{\mu _{k,(1)}\left[ \mathrm{tr}\left( \mathbf{W}_{k,(1)}\mathbf{E}_{k,(1)} \right) -\log _2\left| \mathbf{W}_{k,(1)} \right| \right]}\\
&\mathrm{s}.\mathrm{t}. \left\| \mathbf{F}_{k,\mathrm{u},(1)} \right\| ^2\leqslant p_k\,\,\forall k=1,\dots ,K
\end{aligned}
\end{equation}
is equivalent to the weighted sum-rate maximization problem \eqref{Sum_SE_Fully}, where $\mathbf{W}_{k,(1)}\in \mathbb{C} ^{N\times N}$ is the weight matrix for UE $k$. We notice that \eqref{Sum_MSE} is convex over each optimization variable $\mathbf{F}$, $\mathbf{V}$, $\mathbf{W}$ but is not jointly convex over all optimization variables. Following the method in \cite{5756489}, we can solve \eqref{Sum_MSE} by sequentially fixing two of the three optimization variables $\mathbf{F}$, $\mathbf{V}$, $\mathbf{W}$ and updating the third.

Fixing the other variables, the update of $\mathbf{V}_k$ is given by the MMSE solution as \eqref{eq:MMSE_Combining}. Under the MMSE combining, the MSE matrix is given by \eqref{eq:MSE_Matrix_Optimal_1}. Then, note that optimal $\mathbf{W}_{k,(1)}$ for \eqref{Sum_MSE} is
\begin{equation}\label{W_1}
\mathbf{W}_{k,(1)}^{\mathrm{opt}}=\mathbf{E}_{k,(1)} ^{-1},
\end{equation}
which can be easily derived through the first order optimality condition for $\mathbf{W}_{k,(1)}$ by fixing $\mathbf{F}$ and $\mathbf{V}$.
\begin{rem}
When the MMSE combining $\mathbf{V}_{k}^{\mathrm{MMSE}}$ and $\mathbf{W}_{k,(1)}^{\mathrm{opt}}$ for all UEs are implemented in \eqref{Sum_MSE}, we have $\mathrm{tr}( \mathbf{W}_{k,(1)}\mathbf{E}_{k,(1)}) -\log _2\left| \mathbf{W}_{k,(1)} \right|=\mathrm{tr}\left( \mathbf{I}_N \right)-\log _2| ( \mathbf{E}_{k,(1)}^{\mathrm{opt}}) ^{-1} |$. So the matrix-weighted sum-MSE minimization problem in \eqref{Sum_MSE} would reduce to the equivalent optimization problem of \eqref{Sum_SE_Fully} as\footnote{Note that ``SE'' is equivalent to ``rate'' except from having one scaling factor $(\tau_c-\tau_p)/\tau_c$. Since $\tau_c$ and $\tau_p$ are constants in this paper, so we ignore the difference between SE and rate in the optimization problem.}:
\vspace*{-0.1cm}
\begin{equation}\label{SE_MSE}
\vspace*{-0.1cm}
\begin{aligned}
&\underset{\left\{ \mathbf{F}\right\}}{\max}\sum_{k=1}^K{\mu _{k,(1)}\log _2\left| \left( \mathbf{E}_{k,(1)}^{\mathrm{opt}} \right) ^{-1} \right|}\\
&\mathrm{s}.\mathrm{t}. \left\| \mathbf{F}_{k,\mathrm{u},(1)} \right\| ^2\leqslant p_k\,\,\forall k=1,\dots ,K
\end{aligned}
\end{equation}
which is a well-known relationship between $\mathbf{E}_{k,(1)}^{\mathrm{opt}}$ and $\mathrm{SE}_{k,(1)}^{\mathrm{opt}}$.
\end{rem}

Finally, fixing $\mathbf{V}$ and $\mathbf{W}$, the update of $\mathbf{F}_{k,\mathrm{u},(1)}$ for \eqref{Sum_MSE} results in the optimization problem as\footnote{It is worth mentioning that the updates of optimization variables are based on the preliminary of fixing the other optimization variables. For instance, when updating $\mathbf{F}_{k,\mathrm{u},(1)}$, we should fix the other optimization variables ($\mathbf{V}_k$ and $\mathbf{W}_{k,(1)}$) but not only limited to their respective optimal solutions $\mathbf{V}_{k}^{\mathrm{MMSE}}$ and $\mathbf{W}_{k,(1)}^{\mathrm{opt}}$. So we update $\mathbf{F}_{k,\mathrm{u},(1)}$ based on \eqref{F_Problem} with generalized $\mathbf{V}_k$ and $\mathbf{W}_{k,(1)}$ instead of \eqref{SE_MSE} with optimal $\mathbf{V}_{k}^{\mathrm{MMSE}}$ and $\mathbf{W}_{k,(1)}^{\mathrm{opt}}$.}

\begin{equation}\label{F_Problem}
\begin{aligned}
&\underset{\left\{ \mathbf{F} \right\}}{\min}\sum_{k=1}^K{\mu _{k,(1)}\mathrm{tr}\left[ \mathbf{W}_{k,(1)}\left( \mathbf{I}_N-\mathbf{V}_{k}^{H}\mathbf{\hat{H}}_k\mathbf{F}_{k,\mathrm{u},(1)} \right) \left( \mathbf{I}_N-\mathbf{V}_{k}^{H}\mathbf{\hat{H}}_k\mathbf{F}_{k,\mathrm{u},(1)} \right) ^H \right]}\\
&\,\,   +\sum_{k=1}^K{\mu _{k,(1)}\mathrm{tr}\left[ \mathbf{W}_{k,(1)}\mathbf{V}_{k}^{H}\left( \sum_{l\ne k}^K{\mathbf{\hat{H}}_l\mathbf{F}_{l,\mathrm{u},(1)}\mathbf{F}_{l,\mathrm{u},(1)}^H\mathbf{\hat{H}}_{l}^{H}} \right) \mathbf{V}_k \right]}-\sum_{k=1}^K{\mu _{k,(1)}\log _2\left| \mathbf{W}_{k,(1)} \right|}\\
&\,\,  +\sum_{k=1}^K{\mu _{k,(1)}\mathrm{tr}}\left[ \mathbf{W}_{k,(1)}\mathbf{V}_{k}^{H}\left( \sum_{l=1}^K{\mathbb{E} \left\{ \left. \mathbf{\tilde{H}}_l\mathbf{F}_{l,\mathrm{u},(1)} \mathbf{F}_{l,\mathrm{u},(1)}^H\mathbf{\tilde{H}}_{l}^{H} \right|\mathbf{F} \right\}}+\sigma ^2\mathbf{I}_{ML} \right) \mathbf{V}_k \right]\\
&\mathrm{s}.\mathrm{t}. \left\| \mathbf{F}_{k,\mathrm{u},(1)} \right\| ^2\leqslant p_k\,\,\forall k=1,\dots ,K
\end{aligned}
\end{equation}
which is a convex quadratic optimization problem. So the classic Lagrange multipliers methods and Karush-Kuhn-Tucker (KKT) conditions can be applied to derive an optimal solution. The Lagrange function of \eqref{F_Problem} is
\begin{equation}\label{Lagrange_Function_1}
\begin{aligned}
f\left( \mathbf{F}_{1,\mathrm{u},(1)},\dots ,\mathbf{F}_{K,\mathrm{u},(1)} \right) &=\sum_{k=1}^K{\mu _{k,(1)}}\mathrm{tr}\left[ \mathbf{W}_{k,(1)}\left( \mathbf{I}_N-\mathbf{V}_{k}^{H}\mathbf{\hat{H}}_k\mathbf{F}_{k,\mathrm{u},(1)} \right) \left( \mathbf{I}_N-\mathbf{V}_{k}^{H}\mathbf{\hat{H}}_k\mathbf{F}_{k,\mathrm{u},(1)} \right) ^H \right]\\
&+\sum_{k=1}^K{\mu _{k,(1)}\mathrm{tr}\left[ \mathbf{W}_{k,(1)}\mathbf{V}_{k}^{H}\left( \sum_{l\ne k}^K{\mathbf{\hat{H}}_l\mathbf{F}_{l,\mathrm{u},(1)} \mathbf{F}_{l,\mathrm{u},(1)}^H\mathbf{\hat{H}}_{l}^{H}} \right) \mathbf{V}_k \right]}\\
&+\sum_{k=1}^K{\mu _{k,(1)}\mathrm{tr}\left[ \mathbf{W}_{k,(1)}\mathbf{V}_{k}^{H}\left( \sum_{l=1}^K{\mathbb{E} \left\{ \left. \mathbf{\tilde{H}}_l\mathbf{F}_{l,\mathrm{u},(1)} \mathbf{F}_{l,\mathrm{u},(1)} ^H\mathbf{\tilde{H}}_{l}^{H} \right|\mathbf{F} \right\}}+\sigma ^2\mathbf{I}_{ML} \right) \mathbf{V}_k \right]}\\
&+\sum_{k=1}^K{\lambda _{k,(1)}\left[ \mathrm{tr}\left( \mathbf{F}_{k,\mathrm{u},(1)} \mathbf{F}_{k,\mathrm{u},(1)} ^H \right) -p_k \right]}
\end{aligned}
\end{equation}

Finally, we derive the optimal precoding scheme as the following theorem.
\begin{thm}\label{Optimal_Precoding_1}
By fixing other optimization variables and applying the first-order optimality condition of \eqref{Lagrange_Function_1} with respect to each $\mathbf{F}_{k,\mathrm{u},(1)}$, the optimal precoding scheme is given by
\begin{equation}\label{Optimal_F_1}
\begin{aligned}
\mathbf{F}_{k,\mathrm{u},(1)}^{\mathrm{opt}}&=\mu _{k,(1)}\left[ \sum_{l=1}^K{\mu _{l,(1)}\left( \mathbf{\hat{H}}_{k}^{H}\mathbf{V}_l\mathbf{W}_{l,(1)}\mathbf{V}_{l}^{H}\mathbf{\hat{H}}_k+\mathbb{E} \left\{ \left. \mathbf{\tilde{H}}_{k}^{H}\mathbf{V}_l\mathbf{W}_{l,(1)}\mathbf{V}_{l}^{H}\mathbf{\tilde{H}}_k \right|\mathbf{V},\mathbf{W} \right\} \right)}+\lambda _{k,(1)}\mathbf{I}_N \right] ^{-1}\mathbf{\hat{H}}_{k}^{H}\mathbf{V}_k\mathbf{W}_{k,(1)}\\
&=\mu _{k,(1)}\left[ \sum_{l=1}^K{\mu _{l,(1)}\left( \mathbf{\hat{H}}_{k}^{H}\mathbf{V}_l\mathbf{W}_{l,(1)}\mathbf{V}_{l}^{H}\mathbf{\hat{H}}_k+\mathbf{\bar{C}}_{kl} \right)}+\lambda _{k,(1)}\mathbf{I}_N \right] ^{-1}\mathbf{\hat{H}}_{k}^{H}\mathbf{V}_k\mathbf{W}_{k,(1)},
\end{aligned}
\end{equation}
where $\lambda _{k,(1)}\geqslant 0$ is the Lagrangian multiplier and the $(i,n)$-th element of $\mathbf{\bar{C}}_{kl}\triangleq \mathbb{E} \{  \mathbf{\tilde{H}}_{k}^{H}\mathbf{V}_l\mathbf{W}_{l,(1)}\mathbf{V}_{l}^{H}\mathbf{\tilde{H}}_k |\mathbf{V},\mathbf{W}\}$ $\in \mathbb{C} ^{N\times N}$ is $\left[ \mathbf{\bar{C}}_{kl} \right] _{in}=\mathrm{tr(}\mathbf{\bar{V}}_l\mathbb{E} \{\mathbf{\tilde{h}}_{k,n}\mathbf{\tilde{h}}_{k,i}^{H}\})$ $=\mathrm{tr}\left( \mathbf{\bar{V}}_l\mathbf{C}_{k,in} \right)$ with $
\mathbf{\bar{V}}_l\triangleq \mathbf{V}_l\mathbf{W}_{l,(1)}\mathbf{V}_{l}^{H}$ and $\mathbf{C}_{k,ni}\triangleq \mathbb{E} \{ \mathbf{\tilde{h}}_{k,n}\mathbf{\tilde{h}}_{k,i}^{H}\} =\mathrm{diag}\left( \mathbf{C}_{1k}^{ni},\dots ,\mathbf{C}_{Mk}^{ni} \right)\in \mathbb{C} ^{ML\times ML}$. According to the KKT condition, $\lambda _{k,(1)}$ and $\mathbf{F}_{k,\mathrm{u},(1)}$ should also satisfy
\begin{equation}\label{KKT_1}
\begin{aligned}
\left\| \mathbf{F}_{k,\mathrm{u},(1)} \right\| ^2\leqslant p_k,\quad \lambda _{k,(1)}\left( \left\| \mathbf{F}_{k,\mathrm{u},(1)} \right\| ^2-p_k \right) =0,\quad\lambda _{k,(1)}\geqslant 0.
\end{aligned}
\end{equation}
\end{thm}
\begin{IEEEproof}
The proof is given in Appendix~\ref{Appendix3}.
\end{IEEEproof}


We denote by $\mathbf{F}_{k,\mathrm{u},(1)}(\lambda _{k,(1)})$ the right-hand side of \eqref{Optimal_F_1}, when $\sum_{l=1}^K{\mu _{l,(1)}( \mathbf{\hat{H}}_{k}^{H}\mathbf{V}_l\mathbf{W}_{l,(1)}\mathbf{V}_{l}^{H}\mathbf{\hat{H}}_k+\mathbf{\bar{C}}_{kl} )}$ is invertible and $\mathrm{tr}[ \mathbf{F}_{k,\mathrm{u},(1)}( 0 ) \mathbf{F}_{k,\mathrm{u},(1)}(0) ^H ] \leqslant p_k$, then $\mathbf{F}_{k,\mathrm{u},(1)}^{\mathrm{opt}}=\mathbf{F}_{k,\mathrm{u},(1)}\left( 0 \right)$, otherwise we have
\vspace*{-0.5cm}
\begin{equation}\notag
\begin{aligned}
\mathrm{tr}[ \mathbf{F}_{k,\mathrm{u},(1)}( \lambda _{k,(1)} ) \mathbf{F}_{k,\mathrm{u},(1)}( \lambda _{k,(1)} ) ^H ] =p_k
\end{aligned}
\end{equation}
\vspace*{-0.2cm}
to satisfy \eqref{KKT_1}.
\begin{coro}\label{Decreasing_function_1}
$\mathrm{tr}[ \mathbf{F}_{k,\mathrm{u},(1)}( \lambda _{k,(1)} ) \mathbf{F}_{k,\mathrm{u},(1)}( \lambda _{k,(1)} ) ^H ]$ is a monotonically decreasing function of $\lambda _{k,(1)}$.
\end{coro}
\begin{IEEEproof}
Let $\mathbf{D\Lambda D}^H$ denote the eigendecomposition of $\sum_{l=1}^K{\mu _{l,(1)}( \mathbf{\hat{H}}_{k}^{H}\mathbf{V}_l\mathbf{W}_{l,(1)}\mathbf{V}_{l}^{H}\mathbf{\hat{H}}_k+\mathbf{\bar{C}}_{kl} )}$. Following the method in \cite{5756489}, we define $
\mathbf{\Phi }=\mu _{k,(1)}^2\mathbf{D}^H\mathbf{\hat{H}}_{k}^{H}\mathbf{V}_k \mathbf{W}_{k,(1)}^2\mathbf{\hat{H}}_k\mathbf{V}_{k}^{H}\mathbf{D}$ and we have
\begin{equation}\label{Eigendecomposition_1}
\begin{aligned}
&\mathrm{tr[}\mathbf{F}_{k,\mathrm{u},(1)}(\lambda _{k,(1)})\mathbf{F}_{k,\mathrm{u},(1)}(\lambda _{k,(1)})^H]=\mathrm{tr}\left\{ \left( \mathbf{D\Lambda D}^H+\lambda _{k,(1)}\mathbf{I}_N \right) ^{-1}\mathbf{D\Phi D}^H\left( \mathbf{D\Lambda D}^H+\lambda _{k,(1)}\mathbf{I}_N \right) ^{-1} \right\}\\
&=\mathrm{tr}\left\{ \left( \mathbf{D\Lambda D}^H+\lambda _{k,(1)}\mathbf{I}_N \right) ^{-2}\mathbf{D\Phi D}^H \right\}=\mathrm{tr}\left\{ \left( \mathbf{\Lambda }+\lambda _{k,(1)}\mathbf{I}_N \right) ^{-2} \right\} =\sum_{n=1}^N{\frac{\left[ \mathbf{\Phi } \right] _{nn}}{\left( \left[ \mathbf{\Lambda } \right] _{nn}+\lambda _{k,(1)} \right) ^2}},
\end{aligned}
\end{equation}
so $\mathrm{tr}[ \mathbf{F}_{k,\mathrm{u},(1)}( \lambda _{k,(1)} ) \mathbf{F}_{k,\mathrm{u},(1)}( \lambda _{k,(1)} ) ^H ]$ is a monotonically decreasing function of $\lambda _{k,(1)}$.
\end{IEEEproof}

Based on Corollary~\ref{Decreasing_function_1}, optimum $\lambda _{k,(1)}$ (denoted by $\lambda _{k,(1)}^{\mathrm{opt}}$) can be easily obtained by a one-dimensional (1-D) bisection algorithm so we derive the solution for $\mathbf{F}_{k,\mathrm{u},(1)}( \lambda _{k,(1)}^{\mathrm{opt}} )$. Furthermore, an iterative optimization algorithm for $\mathbf{F}_{k,\mathrm{u},(1)}$, called ``iteratively WMMSE (I-WMMSE) algorithm", is summarized in Algorithm~\ref{algo:iterative1}\footnote{To balance the efficiency and the computational complexity of the proposed algorithm, we also include the stopping criterion ``$R_{(1)}^{\left( i \right)}<R_{(1)}^{\left( i-1 \right)}$. Moreover, the I-WMMSE precoding scheme is derived at iteration $(i-1)$, which may achieve higher sum SE than the one at iteration $i$.}. The convergence of Algorithm \ref{algo:iterative1} is proven in \cite[Theorem 3]{5756489}.
\begin{algorithm}[t]
\label{algo:iterative1}
\caption{I-WMMSE Algorithm for the Design of $\mathbf{F}_{k,\mathrm{u},(1)}$}
\KwIn{Collective channel estimates $\mathbf{\hat{H}}_k$ for all UEs; Estimation error covariance matrices $\mathbf{C}_{ml}$ for all possible pairs; UE weights $\mu _{k,(1)}$ for all UEs;}
\KwOut{Optimal precoding matrices $\mathbf{F}_{k,\mathrm{u},(1)}$ for all UEs ($\mathbf{F}_{k,\mathrm{u},(1)}^{( i)}$ for the first or third stopping criterion and $\mathbf{F}_{k,\mathrm{u},(1)}^{( i-1)}$ for the second stopping criterion);}

{\bf Initiation:} $i=0$, $\mathbf{F}_{k,\mathrm{u},(1)}^{\left( 0 \right)}$ and $R_{(1)}^{\left( 0 \right)}=\sum_{k=1}^K{\mu _{k,(1)}\mathrm{SE}_{k,(1)}^{\left( 0 \right)}}$ for all UEs; maximum iteration number $I_{(1),\max}$ and threshold $\varepsilon_{(1)}$;\\

\Repeat(){$\left| R_{(1)}^{\left( i \right)}-R_{(1)}^{\left( i-1 \right)} \right|/{R_{(1)}^{\left( i-1 \right)}}\leqslant \varepsilon_{(1)}$ or $R_{(1)}^{\left( i \right)}<R_{(1)}^{\left( i-1 \right)}$ or $i\geqslant I_{(1),\max}$}
{
$i=i+1$\\
Update the MMSE combining scheme $\mathbf{V}_{k}^{\left( i \right)}$ with $\mathbf{F}_{k,\mathrm{u},(1)}^{\left( i-1 \right)}$ based on \eqref{eq:MMSE_Combining};\\
Update optimal MSE matrix $\mathbf{E}_{k,(1)}^{\left( i \right)}$ with $\mathbf{F}_{l,\mathrm{u},(1)}^{\left( i-1 \right)}$ based on \eqref{eq:MSE_Matrix_Optimal_1}, and update $\mathbf{W}_{k,(1)}^{\left( i \right)}$ based on \eqref{W_1};\\
Update optimal precoding matrix $\mathbf{F}_{l,\mathrm{u},(1)}^{\left( i\right)}$ with $\mathbf{V}_{k}^{\left( i \right)}$ and $\mathbf{W}_{k,(1)}^{\left( i \right)}$ based on \eqref{Optimal_F_1}, where $\lambda _{k,(1)}^{\left( i \right)}$ is found by a bisection algorithm; \\
Update sum weighted rate $R_{(1)}^{\left( i \right)}=\sum_{k=1}^K{\mu _{k,(1)}\mathrm{SE}_{k,(1)}^{\left( i \right)}}$;\\
}
\end{algorithm}
\vspace{-0.3cm}
\begin{rem}
Note that the design of $\mathbf{F}_{k,\mathrm{p}}$ is a valuable future direction to further improve the system performance. One valuable optimization problem is to minimize the total MSE of the channel estimators of all UEs as
\begin{equation}
\begin{aligned}
&\underset{\{ \mathbf{F}_{k,\mathrm{p}}\}}{\min}\sum_{k=1}^K{\mathrm{tr}\left( \mathbf{C}_k \right)}\\
&\mathrm{s}.\mathrm{t}.\left\| \mathbf{F}_{k,\mathrm{p}} \right\| ^2\leqslant p_k\,\,\forall k=1,\dots ,K
\end{aligned}
\end{equation}
where the optimization goal is only based on the statistical knowledge so $\mathbf{F}_{k,\mathrm{p}}$ is also based on the statistical knowledge.
\end{rem}
\vspace{-0.6cm}
\subsection{Large-Scale Fading Decoding}
\subsubsection{Spectral Efficiency Analysis}
Another promising processing scheme is ``large-scale fading decoding", which is a two-layer decoding scheme to decode the data symbol. Note that UL precoding matrices ($\mathbf{F}_{k,\mathrm{u}}$ and $\mathbf{F}_{k,\mathrm{p}}$) are assumed to be available at all APs and the CPU. In the first layer, AP $m$ applies an arbitrary combining matrix $\mathbf{V}_{mk}\in \mathbb{C} ^{L\times N}$ to derive local detection of $\mathbf{x}_k$ as
\begin{equation}
\begin{aligned}
\mathbf{\tilde{x}}_{mk}
=\mathbf{V}_{mk}^{H}\mathbf{y}_m=\mathbf{V}_{mk}^{H}\mathbf{H}_{mk}\mathbf{F}_{k,\mathrm{u}}\mathbf{x}_k+\sum_{l=1,l\ne k}^K{\mathbf{V}_{mk}^{H}\mathbf{H}_{ml}\mathbf{F}_{l,\mathrm{u}}\mathbf{x}_l}+\mathbf{V}_{mk}^{H}\mathbf{n}_m.
\end{aligned}
\end{equation}

We notice that $\mathbf{V}_{mk}$ is designed based on local channel estimates at AP $m$ and one handy choice is MR combining $\mathbf{V}_{mk}=\mathbf{\hat{H}}_{mk}$. Moreover, local MMSE (L-MMSE) combining
\vspace*{-0.1cm}
\begin{equation}\label{LMMSE_Com}
\vspace*{-0.1cm}
\mathbf{V}_{mk}=\!\! \left( \sum_{l=1}^K{\left(\! \mathbf{\hat{H}}_{ml}\mathbf{\bar{F}}_{l,\mathrm{u}}\mathbf{\hat{H}}_{ml}^{H}+\mathbf{C}_{ml}^{\prime} \! \right)}+\sigma ^2\mathbf{I}_L \!\right) ^{-1}\!\!\mathbf{\hat{H}}_{mk}\mathbf{F}_{k,\mathrm{u}},
\end{equation}
is also regarded as a promising scheme, since \eqref{LMMSE_Com} can minimize $\mathbb{E} \{\parallel \mathbf{x}_k-\mathbf{V}_{mk}^{H}\mathbf{y}_m\parallel ^2|\{ \mathbf{\hat{H}}_{mk}\} ,\{\mathbf{F}_{k,\mathrm{u}} \} \}$.

In the second layer, the ``LSFD'' method is implemented at the CPU \cite{[162]}. The CPU weights all the local estimates $\mathbf{\tilde{x}}_{mk}$ from all APs by the LSFD coefficient matrix 
as
\begin{equation}\label{Data_Final_LSFD}
\begin{aligned}
\mathbf{\hat{x}}_k=\sum_{m=1}^M{\mathbf{A}_{mk}^{H}\mathbf{\tilde{x}}_{mk}}=\sum_{m=1}^M{\mathbf{A}_{mk}^{H}\mathbf{V}_{mk}^{H}\mathbf{H}_{mk}\mathbf{F}_{k,\mathrm{u}}\mathbf{x}_k}+\sum_{m=1}^M{\sum_{l=1,l\ne k}^K{\mathbf{A}_{mk}^{H}\mathbf{V}_{mk}^{H}\mathbf{H}_{ml}\mathbf{F}_{l,\mathrm{u}}\mathbf{x}_l}+}\mathbf{n}_{k}^{\prime},
\end{aligned}
\end{equation}
where $\mathbf{A}_{mk}\in \mathbb{C} ^{N\times N}$ is the complex LSFD coefficient matrix for AP $m$-UE $k$ and $\mathbf{n}_{k}^{\prime}=\sum_{m=1}^M{\mathbf{A}_{mk}^{H}\mathbf{V}_{mk}^{H}\mathbf{n}_m}$. Moreover, we can rewrite $\mathbf{\hat{x}}_k$ in a more compact form as
\begin{equation}\label{Data_Final_LSFD_Compact}
\begin{aligned}
\mathbf{\hat{x}}_k =\mathbf{A}_{k}^{H}\mathbf{G}_{kk}\mathbf{F}_{k,\mathrm{u}}\mathbf{x}_k+\sum_{l=1,l\ne k}^K{\mathbf{A}_{k}^{H}\mathbf{G}_{kl}\mathbf{F}_{l,\mathrm{u}}\mathbf{x}_l}+\mathbf{n}_{k}^{\prime}=\mathbf{A}_{k}^{H}\underset{\mathbf{\tilde{x}}_k}{\underbrace{\left( \mathbf{G}_{kk}\mathbf{F}_{k,\mathrm{u}}\mathbf{x}_k+\sum_{l=1,l\ne k}^K{\mathbf{G}_{kl}\mathbf{F}_{l,\mathrm{u}}\mathbf{x}_l}+\mathbf{\tilde{n}}_{k}^{\prime} \right) }}
\end{aligned}
\end{equation}
where $\mathbf{A}_k\triangleq [ \mathbf{A}_{1k}^{T},\dots ,\mathbf{A}_{Mk}^{T} ] ^T\in \mathbb{C} ^{MN\times N}$, $\mathbf{G}_{kl}\triangleq [ \mathbf{V}_{1k}^{H}\mathbf{H}_{1l};\dots ;\mathbf{V}_{Mk}^{H}\mathbf{H}_{Ml} ] \in \mathbb{C} ^{MN\times N}$ and\\ $\mathbf{\tilde{n}}_{k}^{\prime}=\left[ \mathbf{V}_{1k}^{H}\mathbf{n}_1;\dots ;\mathbf{V}_{Mk}^{H}\mathbf{n}_M \right] \in \mathbb{C} ^{MN\times N}$.

Note that the CPU does not have the knowledge of channel estimates and is only aware of channel statistics \cite{[162]}. The conditional MSE matrix for UE $k$ $\mathbf{E}_{k,(2)}\triangleq \mathbb{E} \left\{ \left( \mathbf{x}_k-\mathbf{\hat{x}}_k \right) ( \mathbf{x}_k-\mathbf{\hat{x}}_k ) ^H\left| \mathbf{\Theta } \right. \right\}$ is
\vspace*{-0.1cm}
\begin{equation}\label{MSE_Matrix}
\begin{aligned}
\mathbf{E}_{k,(2)}&=\mathbf{I}_N-\mathbf{F}_{k,\mathrm{u}}^{H}\mathbb{E} \{ \mathbf{G}_{kk}^{H} \} \mathbf{A}_k-\mathbf{A}_{k}^{H}\mathbb{E}\{ \mathbf{G}_{kk} \} \mathbf{F}_{k,\mathrm{u}}+\mathbf{A}_{k}^{H}\left( \sum_{l=1}^K{\mathbb{E} \{ \mathbf{G}_{kl}\mathbf{\bar{F}}_{l,\mathrm{u}}\mathbf{G}_{kl}^{H} \}}+\sigma ^2\mathbf{S}_k \right) \mathbf{A}_k,
\vspace*{-0.1cm}
\end{aligned}
\end{equation}
where $\mathbf{\Theta }$ denotes all the channel statistics and $\mathbf{S}_k =\mathrm{diag}( \mathbb{E} \{ \mathbf{V}_{1k}^{H}\mathbf{V}_{1k} \} ,\cdots ,\mathbb{E} \{ \mathbf{V}_{Mk}^{H}\mathbf{V}_{Mk} \} )\in \mathbb{C} ^{MN\times MN}$. Then, we apply classical use-and-then-forget bound to obtain the following ergodic achievable SE.

\begin{coro}
For the ``LSFD" scheme, an achievable SE for UE $k$ can be written as
\vspace*{-0.1cm}
\begin{equation}\label{SE_LSFD_Origin}
\begin{aligned}
\mathrm{SE}_{k,(2)}=\left( 1-\frac{\tau _p}{\tau _c} \right) \log _2\left| \mathbf{I}_N+\mathbf{D}_{k,(2)}^{H}\mathbf{\Sigma }_{k,(2)}^{-1}\mathbf{D}_{k,(2)} \right|,
\end{aligned}
\end{equation}
where $ \mathbf{\Sigma }_{k,(2)}=\sum_{l=1}^K{\mathbf{A}_{k}^{H}\mathbb{E} \{ \mathbf{G}_{kl}\mathbf{\bar{F}}_{l,\mathrm{u}}\mathbf{G}_{kl}^{H} \} \mathbf{A}_k}-\mathbf{D}_{k,(2)}\mathbf{D}_{k,(2)}^{H}+\sigma ^2\mathbf{A}_{k}^{H}\mathbf{S}_k\mathbf{A}_k$ and $\mathbf{D}_{k,(2)}=\mathbf{A}_{k}^{H}\mathbb{E} \{ \mathbf{G}_{kk} \} \mathbf{F}_{k,\mathrm{u}}$.
\end{coro}
\begin{proof}
The proof follows similar steps as the proof of \cite[Corollary 2]{04962} and is therefore omitted.
\end{proof}

Note that $\mathbf{A}_k$ can be optimized by the CPU based on channel statistics to maximize the achievable SE in \eqref{SE_LSFD_Origin}. Based on the theory of optimal receivers as in \cite{tse2005fundamentals}, we derive the optimal LSFD coefficient matrix, which not only maximizes the achievable SE but minimizes the conditional MSE, as follows.
\begin{coro}\label{Corollary_Optimal_LSFD}
The achievable SE in \eqref{SE_LSFD_Origin} is maximized by
\vspace*{-0.1cm}
\begin{equation}\label{Optimal_LSFD}
\vspace*{-0.1cm}
\mathbf{A}_{k}^{\mathrm{opt}}=\left( \sum_{l=1}^K{\mathbb{E} \{ \mathbf{G}_{kl}\mathbf{\bar{F}}_{l,\mathrm{u}}\mathbf{G}_{kl}^{H} \}}+\sigma ^2\mathbf{S}_k \right) ^{-1}\mathbb{E} \{ \mathbf{G}_{kk} \} \mathbf{F}_{k,\mathrm{u}},
\end{equation}
leading to the maximum value as
\begin{equation}\label{SE_LSFD_Optimal}
\begin{aligned}
&\mathrm{SE}_{k,(2)}^{\mathrm{opt}}\\
&=\left( 1-\frac{\tau _p}{\tau _c} \right) \log _2 \left| \mathbf{I}_N+\mathbf{F}_{k,\mathrm{u}}^{H}\mathbb{E} \left\{ \mathbf{G}_{kk} \right\} \left( \sum_{l=1}^K{\mathbb{E} \left\{ \mathbf{G}_{kl}\mathbf{\bar{F}}_{l,\mathrm{u}}\mathbf{G}_{kl}^{H} \right\} -\mathbb{E} \left\{ \mathbf{G}_{kk} \right\} \mathbf{\bar{F}}_{k,\mathrm{u}}\mathbb{E} \left\{ \mathbf{G}_{kk}^{H} \right\} +\sigma ^2\mathbf{S}_k} \right) ^{-1}\mathbb{E} \left\{ \mathbf{G}_{kk} \right\} \mathbf{F}_{k,\mathrm{u}} \right|.
\end{aligned}
\end{equation}
Note that the optimal LSFD coefficient matrix in \eqref{Optimal_LSFD} can also minimize the conditional MSE for UE $k$ $\mathrm{MSE}_{k,(2)}=\mathrm{tr}( \mathbf{E}_{k,(2)} )$.
\end{coro}
\begin{proof}
The proof is given in Appendix~\ref{Appendix1}.
\end{proof}

If the optimal LSFD coefficient matrix is applied, the MSE matrix for UE $k$ can be written as
\vspace*{-0.1cm}
\addtocounter{equation}{1}
\begin{equation}\label{MMSE_MSE_Matrix}
\begin{aligned}
\mathbf{E}_{k,(2)}^{\mathrm{opt}}=\mathbf{I}_N-\mathbf{F}_{k,\mathrm{u}}^{H}\mathbb{E} \left\{ \mathbf{G}_{kk}^{H} \right\} \left( \sum_{l=1}^K{\mathbb{E} \{ \mathbf{G}_{kl}\mathbf{\bar{F}}_{l,\mathrm{u}}\mathbf{G}_{kl}^{H} \}}+\sigma ^2\mathbf{S}_k \right) ^{-1}\mathbb{E} \{ \mathbf{G}_{kk} \} \mathbf{F}_{k,\mathrm{u}}.
\end{aligned}
\vspace*{-0.1cm}
\end{equation}
Furthermore, if MR combining $\mathbf{V}_{mk}=\mathbf{\hat{H}}_{mk}$ is applied, we derive closed-form SE expressions as follows.
\begin{thm}\label{Th_Closed_Form}
For MR combining $\mathbf{V}_{mk}=\mathbf{\hat{H}}_{mk}$, \eqref{SE_LSFD_Origin} can be computed in closed-form as
\begin{equation}\label{SE_Closedform}
\mathrm{SE}_{k,(2),\mathrm{c}}=\left( 1-\frac{\tau _p}{\tau _c} \right) \log _2\left| \mathbf{I}_N+\mathbf{D}_{k,(2),\mathrm{c}}^{H}\mathbf{\Sigma }_{k,(2),\mathrm{c}}^{-1}\mathbf{D}_{k,(2),\mathrm{c}} \right|,
\end{equation}
where $\mathbf{\Sigma }_{k,(2),\mathrm{c}}=\mathbf{A}_{k}^{H}( \sum_{l=1}^K{\mathbf{T}_{kl,( 1 )}+\sum_{l\in \mathcal{P} _k}^{}{\mathbf{T}_{kl,( 2 )}}} ) \mathbf{A}_k-\mathbf{D}_{k,(2),\mathrm{c}}\mathbf{D}_{k,(2),\mathrm{c}}^{H}+\sigma ^2\mathbf{A}_{k}^{H}\mathbf{S}_{k,\mathrm{c}}\mathbf{A}_k$ and $\mathbf{D}_{k,(2),\mathrm{c}}=\mathbf{A}_{k}^{H}\mathbf{Z}_k\mathbf{F}_{k,\mathrm{u}}$, with $\mathbb{E} \{ \mathbf{G}_{kk} \} =\mathbf{Z}_k=[ \mathbf{Z}_{1k}^{T},\dots ,\mathbf{Z}_{Mk}^{T} ] ^T$ and $\mathbf{S}_{k,\mathrm{c}}=\mathrm{diag}( \mathbf{Z}_{1k},\cdots ,\mathbf{Z}_{Mk} )$ with the $\left( n,n^{\prime} \right) $-th element of $\mathbf{Z}_{mk}\in \mathbb{C} ^{N\times N}$ being $\left[ \mathbf{Z}_{mk} \right] _{nn^{\prime}}=\mathrm{tr}( \mathbf{\hat{R}}_{mk}^{n^{\prime}n} )$. Moreover, $\mathbf{T}_{kl,\left( 1 \right)}\triangleq \mathrm{diag}( \mathbf{\Gamma }_{kl,1}^{( 1 )},\cdots ,\mathbf{\Gamma }_{kl,M}^{( 1 )} ) \in \mathbb{C} ^{MN\times MN}$ and $\mathbf{T}_{kl,\left( 2 \right)}^{mm^{\prime}}=\mathbf{\Gamma }_{kl,m}^{\left( 2 \right)}-\mathbf{\Gamma }_{kl,m}^{\left( 1 \right)}$ if $m=m^{\prime}$ and $\mathbf{\Lambda }_{mkl}\mathbf{\bar{F}}_{l,\mathrm{u}}\mathbf{\Lambda }_{m^{\prime}lk}$ otherwise, where $\mathbf{T}_{kl,\left( 2 \right)}^{mm^{\prime}}$ denotes $\left( m,m^{\prime}\right) $-submatrix of $\mathbf{T}_{kl,\left( 2 \right)}\in \mathbb{C} ^{MN\times MN}$, the $\left( n,n^{\prime} \right) $-th element of $N\times N$-dimension complex matrices $\mathbf{\Lambda }_{mkl}$, $\mathbf{\Lambda }_{m^{\prime}lk}$, $\mathbf{\Gamma }_{kl,m}^{\left( 1 \right)}$ and $\mathbf{\Gamma }_{kl,m}^{\left( 2 \right)}$ are $[ \mathbf{\Lambda }_{mkl} ] _{nn^{\prime}}=\mathrm{tr}( \mathbf{\Xi }_{mkl}^{n^{\prime}n} ) $, $[ \mathbf{\Lambda }_{m^{\prime}lk} ] _{nn^{\prime}}=\mathrm{tr}( \mathbf{\Xi }_{m^{\prime}lk}^{n^{\prime}n} ) $, $[ \mathbf{\Gamma }_{mkl}^{( 1 )} ] _{nn^{\prime}}=\sum_{i=1}^N{\sum_{i^{\prime}=1}^N{[ \mathbf{\bar{F}}_{l,\mathrm{u}} ] _{i^{\prime}i}\mathrm{tr}( \mathbf{R}_{ml}^{i^{\prime}i}\mathbf{\hat{R}}_{mk}^{n^{\prime}n} )}}$ and $[ \mathbf{\Gamma }_{kl,m}^{( 2 )} ] _{nn^{\prime}}$ given by
\begin{equation}\label{eq:Gamma_2}
\begin{aligned}
\left[ \mathbf{\Gamma }_{kl,m}^{\left( 2 \right)} \right] _{nn^{\prime}}&=\sum_{i=1}^N{\sum_{i^{\prime}=1}^N{\left[ \mathbf{\bar{F}}_l \right] _{i^{\prime}i}\left\{ \mathrm{tr}\left( \mathbf{R}_{ml}^{i^{\prime}i}\mathbf{P}_{mkl,\left( 1 \right)}^{n^{\prime}n} \right) \right.}}\\
&\left. +\tau _{p}^{2}\sum_{q_1=1}^N{\sum_{q_2=1}^N{\left[ \mathrm{tr}\left( \mathbf{\tilde{P}}_{mkl,\left( 2 \right)}^{q_1n}\mathbf{\tilde{R}}_{ml}^{i^{\prime}q_2}\mathbf{\tilde{R}}_{ml}^{q_2i}\mathbf{\tilde{P}}_{mkl,\left( 2 \right)}^{n^{\prime}q_1} \right) +\mathrm{tr}\left( \mathbf{\tilde{P}}_{mkl,\left( 2 \right)}^{q_1n}\mathbf{\tilde{R}}_{ml}^{i^{\prime}q_2} \right) \mathrm{tr}\left( \mathbf{\tilde{P}}_{mkl,\left( 2 \right)}^{n^{\prime}q_2}\mathbf{\tilde{R}}_{ml}^{q_2i} \right) \right]}} \right\}
\end{aligned}
\end{equation}
with $\mathbf{\Xi }_{mkl}=\tau _p\mathbf{R}_{ml}\mathbf{\tilde{F}}_{l,\mathrm{p}}^{H}\mathbf{\Psi }_{mk}^{-1}\mathbf{\tilde{F}}_{k,\mathrm{p}}\mathbf{R}_{mk}$, $\mathbf{\Xi }_{m^{\prime}lk}=\tau _p\mathbf{R}_{m^{\prime}k}\mathbf{\tilde{F}}_{k,\mathrm{p}}^{H}\mathbf{\Psi }_{m^{\prime}k}^{-1}\mathbf{\tilde{F}}_{l,\mathrm{p}}\mathbf{R}_{m^{\prime}l}$, $\mathbf{P}_{mkl,( 1 )}=\tau _p\mathbf{S}_{mk}( \mathbf{\Psi }_{mk}-\tau _p\mathbf{\tilde{F}}_{l,\mathrm{p}}\mathbf{R}_{ml}\mathbf{\tilde{F}}_{l,\mathrm{p}}^{H} ) \mathbf{S}_{mk}^{H}$, $\mathbf{S}_{mk}=\mathbf{R}_{mk}\mathbf{\tilde{F}}_{k,\mathrm{p}}^{H}\mathbf{\Psi }_{mk}^{-1}$, $\mathbf{P}_{mkl,( 2 )}=\mathbf{S}_{mk}\mathbf{\tilde{F}}_{l,\mathrm{p}}\mathbf{R}_{ml}\mathbf{\tilde{F}}_{l,\mathrm{p}}^{H}\mathbf{S}_{mk}^{H}$, $\mathbf{\tilde{R}}_{ml}^{ni}$ and $\mathbf{\tilde{P}}_{mkl,( 2 )}^{ni}$ being $( n,i )$-submatrix of $\mathbf{R}_{ml}^{\frac{1}{2}}$ and $\mathbf{P}_{mkl,( 2 )}^{\frac{1}{2}}$, respectively. Furthermore, the optimal LSFD coefficient matrix in \eqref{Optimal_LSFD} and MSE matrix in \eqref{MMSE_MSE_Matrix} can also be computed in closed-form as
\vspace*{-0.5cm}
\begin{equation}\label{Closed_form_LSFD_MSE}
\begin{aligned}
\begin{cases}
	\mathbf{A}_{k,\mathrm{c}}^{\mathrm{opt}}=\left( \sum_{l=1}^K{\mathbf{T}_{kl,\left( 1 \right)}+\sum_{l\in \mathcal{P} _k}^{}{\mathbf{T}_{kl,\left( 2 \right)}}}+\sigma ^2\mathbf{S}_{k,\mathrm{c}} \right) ^{-1}\mathbf{Z}_k\mathbf{F}_{k,\mathrm{u}},\\
	\mathbf{E}_{k,(2),\mathrm{c}}^{\mathrm{opt}}=\mathbf{I}_N-\mathbf{F}_{k,\mathrm{u}}^{H}\mathbf{Z}_{k}^{H}\left( \sum_{l=1}^K{\mathbf{T}_{kl,\left( 1 \right)}+\sum_{l\in \mathcal{P} _k}^{}{\mathbf{T}_{kl,\left( 2 \right)}}}+\sigma ^2\mathbf{S}_{k,\mathrm{c}} \right) ^{-1}\mathbf{Z}_k\mathbf{F}_{k,\mathrm{u}}.\\
\end{cases}
\end{aligned}
\end{equation}
\end{thm}
\begin{IEEEproof}
The proof is given in Appendix~\ref{Appendix2}.
\end{IEEEproof}

\subsubsection{Iteratively WMMSE Precoding Design}
For the LSFD scheme, we also investigate a weighted sum-rate maximization problem as
\begin{equation}\label{Sum_SE_2}
\begin{aligned}
&\underset{\left\{ \mathbf{F} \right\}}{\max}\sum_{k=1}^K{\mu_{k,(2)}\mathrm{SE}_{k,(2)}}\\
&\mathrm{s}.\mathrm{t}. \left\| \mathbf{F}_{k,\mathrm{u},(2)} \right\| ^2\leqslant p_k\,\,\forall k=1,\dots ,K
\end{aligned}
\end{equation}
where $\mu_{k,(2)}$ represents the priority weight of UE $k$ for the ``LSFD" scheme and $\mathrm{SE}_{k,(2)}$ is given in \eqref{SE_LSFD_Origin} with arbitrary combining structure in the first decoding layer.

Similarly, the matrix-weighted sum-MSE minimization problem as\footnote{The notation $\mathbf{G}$ denotes all $\mathbf{G}$-relevant variables, like $\mathbb{E} \{ \mathbf{G}_{kl}\mathbf{\bar{F}}_{l,\mathrm{u},(2)}\mathbf{G}_{kl}^{H} \} $ and $\mathbb{E} \{ \mathbf{G}_{kk} \}$, etc.}
\vspace*{-0.1cm}
\begin{equation}\label{Sum_MSE_2}
\vspace*{-0.1cm}
\begin{aligned}
&\underset{\left\{ \mathbf{F},\mathbf{A},\mathbf{W},\mathbf{G},\mathbf{S} \right\}}{\min}\sum_{k=1}^K{\mu _{k,(2)}}\left[ \mathrm{tr}\left( \mathbf{W}_{k,(2)}\mathbf{E}_{k,(2)} \right) -\log _2\left| \mathbf{W}_{k,(2)} \right| \right]\\
&\mathrm{s}.\mathrm{t}. \left\| \mathbf{F}_{k,\mathrm{u},(2)} \right\| ^2\leqslant p_k\,\,\forall k=1,\dots ,K
\end{aligned}
\end{equation}
is equivalent to the weighted sum-rate maximization problem \eqref{Sum_SE_2}, where $\mathbf{W}_{k,(2)}$ is the weight matrix for UE $k$. Note that \eqref{Sum_MSE_2} is convex over each optimization variable $\mathbf{F}$, $\mathbf{A}$, $\mathbf{W}$, $\mathbf{G}$, $\mathbf{S}$ but is not jointly convex over all optimization variables. So we can solve \eqref{Sum_MSE_2} by sequentially fixing four of the five optimization variables $\mathbf{F}$, $\mathbf{A}$, $\mathbf{W}$, $\mathbf{G}$, $\mathbf{S}$ and updating the fifth.\footnote{As for $\mathbf{G}$ and $\mathbf{S}$, if L-MMSE combining scheme applied, $\mathbb{E} \left\{ \mathbf{G}_{kk} \right\}$ and $\mathbf{S}_k$ are relevant to $\mathbf{F}_{k,\mathrm{u},(2)}$ so we should also update them. On the contrary, $\mathbb{E} \{ \mathbf{G}_{kk} \}$ and $\mathbf{S}_k$ with MR combining structure are irrelevant to $\mathbf{F}$ so we only need to update $\mathbb{E} \{ \mathbf{G}_{kl}\mathbf{\bar{F}}_{l,\mathrm{u},(2)}\mathbf{G}_{kl}^{H} \} $.}

The update of $\mathbf{A}_k$ and $\mathbf{E}_{k,(2)}$ are given by the optimal LSFD scheme \eqref{Optimal_LSFD} and MSE matrix with optimal LSFD scheme \eqref{MMSE_MSE_Matrix}. Note that optimal $\mathbf{W}_{k,(2)}$ for \eqref{Sum_MSE_2} is $\mathbf{W}_{k,(2)}^{\mathrm{opt}}=\mathbf{E}_{k,(2)}^{-1}$.
\begin{rem}
When $\mathbf{A}_{k}^{\mathrm{opt}}$ and $\mathbf{W}_{k,(2)}^{\mathrm{opt}}$ for all UEs are applied in \eqref{Sum_MSE_2}, we notice that \eqref{Sum_MSE_2} becomes to the equivalent optimization problem of \eqref{Sum_SE_2} as
\vspace*{-0.1cm}
\begin{equation}\label{SE_MSE_2}
\vspace*{-0.1cm}
\begin{aligned}
&\underset{\left\{ \mathbf{F},\mathbf{G},\mathbf{S} \right\}}{\max}\sum_{k=1}^K{\mu _{k,(2)}\log _2\left| \left( \mathbf{E}_{k,(2)}^{\mathrm{opt}} \right) ^{-1} \right|}\\
&\mathrm{s}.\mathrm{t}. \left\| \mathbf{F}_{k,\mathrm{u},(2)} \right\| ^2\leqslant p_k\,\,\forall k=1,\dots ,K
\end{aligned}
\end{equation}which is a well-known relationship between $\mathbf{E}_{k,(2)}^{\mathrm{opt}}$ and $\mathrm{SE}_{k,(2)}^{\mathrm{opt}}$ and proven in Appendix~\ref{MSE_SE}.
\end{rem}

Last but not least, fixing other variables, the update of $\mathbf{F}_{k,\mathrm{u},(2)}$ for \eqref{Sum_MSE_2} results in the optimization problem as
\begin{equation}\label{F_Problem_2}
\begin{aligned}
&\underset{\left\{ \mathbf{F} \right\}}{\min}\sum_{k=1}^K{\mu _{k,(2)}\left[ \mathrm{tr}\left( \mathbf{W}_k\left( \mathbf{I}_N-\mathbf{F}_{k,\mathrm{u},(2)}^{H}\mathbb{E} \left\{ \mathbf{G}_{kk}^{H} \right\} \mathbf{A}_k \right) \left( \mathbf{I}_N-\mathbf{F}_{k,\mathrm{u},(2)}^{H}\mathbb{E} \left\{ \mathbf{G}_{kk}^{H} \right\} \mathbf{A}_k \right) ^H \right) \right]}\\
&+\sum_{k=1}^K{\mu _{k,(2)}\left[ \mathrm{tr}\left( \mathbf{W}_{k,(2)}\mathbf{A}_{k}^{H}\left( \sum_{l\ne k}^K{\mathbb{E} \left\{ \mathbf{G}_{kl}\mathbf{\bar{F}}_{l,\mathrm{u},(2)}\mathbf{G}_{kl}^{H} \right\}}+\sigma ^2\mathbf{S}_k \right) \mathbf{A}_k \right) \right]}\\ &\mathrm{s}.\mathrm{t}. \left\| \mathbf{F}_{k,\mathrm{u},(2)} \right\| ^2\leqslant p_k\,\,\forall k=1,\dots ,K
\end{aligned}
\end{equation}
which is a convex quadratic optimization problem. Thus, we can also derive the optimal precoding scheme by applying classic Lagrange multipliers methods and KKT conditions. The Lagrange function of \eqref{F_Problem_2} is
\begin{equation}\label{Lagrange_Function_2}
\begin{aligned}
f\left( \mathbf{F}_{1,\mathrm{u},(2)},\dots ,\mathbf{F}_{K,\mathrm{u},(2)} \right) &=\sum_{k=1}^K{\mu _{k,(2)}\left[ \mathrm{tr}\left( \mathbf{W}_{k,(2)}\left( \mathbf{I}_N-\mathbf{F}_{k,\mathrm{u},(2)}^{H}\mathbb{E} \left\{ \mathbf{G}_{kk}^{H} \right\} \mathbf{A}_k \right) \left( \mathbf{I}_N-\mathbf{F}_{k,\mathrm{u},(2)}^{H}\mathbb{E} \left\{ \mathbf{G}_{kk}^{H} \right\} \mathbf{A}_k \right) ^H \right) \right]}\\
&+\sum_{k=1}^K{\mu _{k,(2)}\left[ \mathrm{tr}\left( \mathbf{W}_{k,(2)}\mathbf{A}_{k}^{H}\left( \sum_{l\ne k}^K{\mathbb{E} \left\{ \mathbf{G}_{kl}\mathbf{\bar{F}}_{l,\mathrm{u},(2)}\mathbf{G}_{kl}^{H} \right\}}+\sigma ^2\mathbf{S}_k \right) \mathbf{A}_k \right) \right]}\\
&+\sum_{k=1}^K{\lambda _{k,(2)}\left( \mathrm{tr}\left( \mathbf{F}_{k,\mathrm{u},(2)}\mathbf{F}_{k,\mathrm{u},(2)}^{H} \right) -p_k \right)}.
\end{aligned}
\end{equation}

\begin{thm}\label{Optimal_Precoding_2}
By applying the first-order optimality condition of \eqref{Lagrange_Function_2} with respect to each $\mathbf{F}_{k,\mathrm{u},(2)}$ and fixing other optimization variables, we obtain the optimal precoding scheme as
\begin{equation}
\begin{aligned}\label{Optimal_F_2}
\mathbf{F}_{k,\mathrm{u},(2)}^{\mathrm{opt}}=\mu _{k,(2)}\left( \sum_{l=1}^K{\mu _{l,(2)}\mathbb{E} \left\{ \mathbf{G}_{lk}^{H}\mathbf{A}_l\mathbf{E}_{l,(2)}^{-1}\mathbf{A}_{l}^{H}\mathbf{G}_{lk} \right\}}+ \lambda _{k,(2)}\mathbf{I}_N \right) ^{-1}\mathbb{E} \left\{ \mathbf{G}_{kk}^{H} \right\} \mathbf{A}_k\mathbf{E}_{k,(2)}^{-1},
\end{aligned}
\end{equation}
where $\lambda _{k,(2)}\geqslant 0$ is the Lagrangian multiplier during the phase of ``LSFD" scheme. According to the KKT condition, $\lambda _{k,(2)}$ and $\mathbf{F}_{k,\mathrm{u},(2)}$ should also satisfy
\begin{equation}\label{KKT_2}
\begin{aligned}
\left\| \mathbf{F}_{k,\mathrm{u},(2)} \right\| ^2\leqslant p_k,\quad\lambda _{k,(2)}\left( \left\| \mathbf{F}_{k,\mathrm{u},(2)} \right\| ^2-p_k \right) =0,\quad\lambda _{k,(2)}\geqslant 0.
\end{aligned}
\end{equation}
\end{thm}
Note that when $\sum_{l=1}^K{\mu _{l,(2)}\mathbb{E} \{ \mathbf{G}_{lk}^{H}\mathbf{A}_l\mathbf{E}_{l,(2)}^{-1}\mathbf{A}_{l}^{H}\mathbf{G}_{lk} \}}$ is invertible and $\mathrm{tr}\left[ \mathbf{F}_{k,\mathrm{u},(2)}( 0 ) \mathbf{F}_{k,\mathrm{u},(2)}( 0 \right) ^H ] \leqslant p_k$, then $\mathbf{F}_{k,\mathrm{u},(2)}^{\mathrm{opt}}=\mathbf{F}_{k,\mathrm{u},(2)}\left( 0 \right) $, otherwise we must have $\mathrm{tr}[ \mathbf{F}_{k,\mathrm{u},(2)}( \lambda _{k,(2)} ) \mathbf{F}_{k,\mathrm{u,(2)}}( \lambda _{k,(2)} ) ^H ] =p_k$. Following the similar method in Corollary~\ref{Decreasing_function_1}, we notice that $\lambda _{k,(2)}$ can be easily found by a 1-D bisection algorithm since $\mathrm{tr}[ \mathbf{F}_{k,\mathrm{u},(2)}( \lambda _{k,(2)} ) \mathbf{F}_{k,\mathrm{u},(2)}( \lambda _{k,(2)} ) ^H ] $ is a monotonically decreasing function of $\lambda _{k,(2)}$.

Moreover, if MR combining $\mathbf{V}_{mk}=\mathbf{\hat{H}}_{mk}$ is applied in the first layer, we can compute expectations in \eqref{Optimal_F_2} in closed-form as following theorem.
\begin{thm}\label{F_Th_Closed_Form}
With MR combining $\mathbf{V}_{mk}=\mathbf{\hat{H}}_{mk}$ and the optimal LSFD scheme applied, we can compute $\mathbb{E} \{ \mathbf{G}_{kk}^{H} \}$, $\mathbf{A}_k^{\mathrm{opt}}$, and $\mathbf{E}_{k,(2)}^{\mathrm{opt}}$ in closed-form as Theorem~\ref{Th_Closed_Form}. Moreover, we have $\mathbf{\bar{T}}_{lk}=\mathbb{E} \{\mathbf{G}_{lk}^{H}\mathbf{A}_l\mathbf{E}_{l,(2)}^{-1}\mathbf{A}_{l}^{H}\mathbf{G}_{lk}\}\in\mathbb{C} ^{N\times N}$ where the $(i,n)$-th element of $\mathbf{\bar{T}}_{lk}$ is $\mathrm{tr}( \mathbf{\bar{A}}_l\mathbf{\bar{G}}_{lk,ni} ) $ with $\mathbf{\bar{A}}_l\triangleq \mathbf{A}_l\mathbf{E}_{l,(2)}^{-1}\mathbf{A}_{l}^{H}$ and the $[ \left( m-1 \right) N+p,\left( m^{\prime}-1 \right) N+p^{\prime}] $-th (or $[o,j]$-th briefly) entry of $\mathbf{\bar{G}}_{lk,ni}\triangleq \mathbb{E} \{ \mathbf{g}_{lk,n}\mathbf{g}_{lk,i}^{H} \} \in \mathbb{C} ^{MN\times MN}$ being
\begin{equation}\label{gg}
\begin{aligned}
\mathbb{E} \{\mathbf{g}_{lk,n}\mathbf{g}_{lk,i}^{H}\}_{oj}=\begin{cases}
	0,\quad l\notin \mathcal{P} _k,m\ne m^{\prime}\\
	\mathrm{tr(}\mathbf{R}_{mk}^{ni}\mathbf{\hat{R}}_{ml}^{p^{\prime}p}),\quad l\notin \mathcal{P} _k,m=m^{\prime}\\
	\mathrm{tr(}\mathbf{\Xi }_{mlk}^{np})\mathrm{tr(}\mathbf{\Xi }_{m^{\prime}kl}^{p^{\prime}i}),\quad l\in \mathcal{P} _k,m\ne m^{\prime}\\
	\mathrm{tr}\left( \mathbf{R}_{mk}^{ni}\mathbf{P}_{mlk,\left( 1 \right)}^{p^{\prime}p} \right) +\tau _{p}^{2}\sum_{q_1=1}^N{\sum_{q_2=1}^N{\mathrm{tr}\left( \mathbf{\tilde{P}}_{mlk,\left( 2 \right)}^{q_1p}\mathbf{\tilde{R}}_{mk}^{nq_2}\mathbf{\tilde{R}}_{mk}^{q_2i}\mathbf{\tilde{P}}_{mlk,\left( 2 \right)}^{p^{\prime}q_1} \right)}}\\
	\quad \quad\quad +\tau _{p}^{2}\sum_{q_1=1}^N{\sum_{q_2=1}^N{\mathrm{tr}\left( \mathbf{\tilde{P}}_{mlk,\left( 2 \right)}^{q_1n}\mathbf{\tilde{R}}_{mk}^{nq_1} \right) \mathrm{tr}\left( \mathbf{\tilde{P}}_{mlk,\left( 2 \right)}^{p^{\prime}q_2}\mathbf{\tilde{R}}_{mk}^{q_2i} \right)}}, \quad l\in \mathcal{P} _k,m=m^{\prime}\quad\\
\end{cases}
\end{aligned}
\end{equation}
where $\mathbf{\Xi }_{mlk}=\tau _p\mathbf{R}_{mk}\mathbf{\tilde{F}}_{k,\mathrm{p}}^{H}\mathbf{\Psi }_{mk}^{-1}\mathbf{\tilde{F}}_{l,\mathrm{p}}\mathbf{R}_{ml}$, $\mathbf{\Xi }_{m^{\prime}kl}=\tau _p\mathbf{R}_{m^{\prime}l}\mathbf{\tilde{F}}_{l,\mathrm{p}}^{H}\mathbf{\Psi }_{m^{\prime}l}^{-1}\mathbf{\tilde{F}}_{k,\mathrm{p}}\mathbf{R}_{m^{\prime}k}$, $\mathbf{S}_{ml}=\mathbf{R}_{ml}\mathbf{\tilde{F}}_{l,\mathrm{p}}^{H}\mathbf{\Psi }_{ml}^{-1}$, $\mathbf{P}_{mlk,\left( 1 \right)}=\tau _p\mathbf{S}_{ml}( \mathbf{\Psi }_{ml}-\tau _p\mathbf{\tilde{F}}_{k,\mathrm{p}}\mathbf{R}_{mk}\mathbf{\tilde{F}}_{k,\mathrm{p}}^{H} ) \mathbf{S}_{ml}^{H}$ and $\mathbf{P}_{mlk,\left( 2 \right)}=\mathbf{S}_{ml}\mathbf{\tilde{F}}_{k,\mathrm{p}}\mathbf{R}_{mk}\mathbf{\tilde{F}}_{k,\mathrm{p}}^{H}\mathbf{S}_{ml}^{H}$. Plugging the derived results into \eqref{Optimal_F_2}, we can compute $\mathbf{F}_{k,\mathrm{u},(2)}^{\mathrm{opt}}$ in closed-form as
\begin{equation}\label{F_2_Closed_form}
\begin{aligned}
\mathbf{F}_{k,\mathrm{u},(2),\mathrm{c}}^{\mathrm{opt}}=\mu _{k,(2)}\left( \sum_{l=1}^K{\mu _{l,(2)} \mathbf{\bar{T}}_{lk} }+ \lambda _{k,(2)}\mathbf{I}_N \right) ^{-1} \mathbf{Z}_{k}^{H}  \mathbf{A}_{k,\mathrm{c}}^{\mathrm{opt}}\mathbf{E}_{k,(2),\mathrm{c}}^{\mathrm{opt},-1}.
\end{aligned}
\end{equation}
\end{thm}
\begin{IEEEproof}
The proof is given in Appendix~\ref{F_Closed_Form}.
\end{IEEEproof}

Furthermore, an iterative optimization algorithm for $\mathbf{F}_{k,\mathrm{u},(2)}$ is summarized in Algorithm \ref{algo:iterative}. The convergence of Algorithm \ref{algo:iterative} is proven in \cite[Theorem 3]{5756489}.
\vspace*{-0.4cm}
\begin{rem}
Relying on the iterative minimization of weighted MSE, two efficient uplink I-WMMSE precoding schemes to maximize the weighted sum SE are proposed. The I-WMMSE precoding schemes for the ``FCP" and ``LSFD" schemes are investigated in Algorithm~\ref{algo:iterative1} and Algorithm~\ref{algo:iterative}, respectively. Note that the design of I-WMMSE precoding scheme for the FCP/LSFD is based on instantaneous CSI/channel statistics, respectively. More importantly, we can compute I-WMMSE precoding schemes in novel closed-form only for the LSFD scheme with MR combining based on Theorm~\ref{F_Th_Closed_Form}.
\end{rem}

\begin{algorithm}[t]
\label{algo:iterative}
\caption{I-WMMSE Algorithm for the Design of $\mathbf{F}_{k,\mathrm{u},(2)}$}
\KwIn{Channel statistics $\mathbf{\Theta }$ for all possible pairs; UE weights $\mu _{k,(2)}$ for all UEs;}
\KwOut{Optimal precoding matrices $\mathbf{F}_{k,\mathrm{u},(2)}$ for all UEs ($\mathbf{F}_{k,\mathrm{u},(2)}^{( i)}$ for the first or third stopping criterion and $\mathbf{F}_{k,\mathrm{u},(2)}^{( i-1)}$ for the second stopping criterion);}

{\bf Initiation:} $i=0$, $\mathbf{F}_{k,\mathrm{u},(2)}^{\left( 0 \right)}$ and $R_{(2)}^{\left( 0 \right)}=\sum_{k=1}^K{\mu _{k,(2)}\mathrm{SE}_{k,(2)}^{\left( 0 \right)}}$ for all UEs; maximum iteration number $I_{(2),\max}$ and threshold $\varepsilon_{(2)} $;\\

\Repeat(){$| R_{(2)}^{\left( i \right)}-R_{(2)}^{\left( i-1 \right)} |/{R_{(2)}^{\left( i-1 \right)}}\leqslant \varepsilon_{(2)} $ or $R_{(2)}^{\left( i \right)}<R_{(2)}^{\left( i-1 \right)}$ or $i\geqslant I_{(2),\max}$}
{
$i=i+1$\\
Update channel statistics $\mathbf{\Theta }^{\left( i \right)}$, such as $\mathbb{E} \{ \mathbf{G}_{kk}^{\left( i \right)} \} $, $\mathbb{E} \{ \mathbf{G}_{kl}^{\left( i \right)}\mathbf{\bar{F}}_{l,\mathrm{u}}^{\left( i-1 \right)}( \mathbf{G}_{kl}^{\left( i \right)} ) ^H \}$ and $\mathbf{S}_{k}^{\left( i \right)}$;\\
Update optimal LSFD matrix $\mathbf{A}_{k}^{\left( i \right)}$ with $\mathbf{F}_{l,\mathrm{u},(2)}^{\left( i-1 \right)}$ and $\mathbf{\Theta }^{\left( i \right)}$ based on \eqref{Optimal_LSFD};\\
Update optimal MSE matrix $\mathbf{E}_{k,(2)}^{\left( i \right)}$ with $\mathbf{F}_{l,\mathrm{u},(2)}^{\left( i-1 \right)}$, $\mathbf{A}_{k}^{\left( i \right)}$ and $\mathbb{E} \{ \mathbf{G}_{kk}^{\left( i \right)} \} $ based on \eqref{MMSE_MSE_Matrix} and update $\mathbf{W}_{k,(2)}^{\left( i \right)}$;\\
Update optimal precoding matrix $\mathbf{F}_{k,\mathrm{u},(2)}^{\left( i \right)}$ with $\mathbf{A}_{k}^{\left( i \right)}$, $\mathbf{W}_{k,(2)}^{\left( i \right)}$ and $\mathbf{\Theta }^{\left( i \right)}$ based on \eqref{Optimal_F_2}, where $\lambda _{k,(2)}^{\left( i \right),}$ is found by a bisection algorithm; \\
Update sum weighted rate $R_{(2)}^{\left( i \right)}=\sum_{k=1}^K{\mu _{k,(2)}\mathrm{SE}_{k,(2)}^{\left( i \right)}}$;\\
}
\end{algorithm}

\vspace*{-0.6cm}
\section{Precoding Implementation and Complexity Analysis}\label{Precoding Implementation}
In this section, we discuss the practical implementation and analyze computational complexity for the UL precoding schemes investigated in Section~\ref{sec:Iterative Optimization}.
\subsection{Precoding Implementation}
\vspace{-0.3cm}
\subsubsection{Precoding Characteristics}

As described above, we investigate a standard block fading model, where the channel response is constant and frequency flat in a coherence block, which contains $\tau_c$ channel uses. For the ``fully centralized processing" scheme, we notice that the I-WMMSE precoding design is implemented at the CPU based on the instantaneous CSI as \eqref{Optimal_F_1}. Moreover, to guarantee the convergence of Algorithm~\ref{algo:iterative1}, only MMSE combining as \eqref{eq:MMSE_Combining} is advocated to detect the UL data since the equivalent relationship between $\mathbf{E}_{k,(1)}^{\mathrm{opt}}$ and $\mathrm{SE}_{k,(1)}^{\mathrm{opt}}$, which only satisfies with MMSE combining, should be guaranteed. As for the LSFD scheme, the optimal design of $\mathbf{F}_{k,\mathrm{u},(2)}^{\mathrm{opt}}$ as \eqref{Optimal_F_2} can only be implemented at the CPU, but relies only on channel statistics. Besides, L-MMSE or MR combining can be applied at each AP. When MR combining is applied, all terms in Algorithm~\ref{algo:iterative} can be computed in closed-form as Theorem~\ref{F_Th_Closed_Form}.\\
\vspace*{-0.7cm}
\subsubsection{Fronthaul Requirements}
For the FCP scheme with the I-WMMSE precoding, in each coherence block, all APs should relay their received signals to the CPU and the CPU requires precoding matrices $\mathbf{F}_{k,\mathrm{u},(1)}$ feedback to all UEs. All APs need to send $\tau _cML$ complex scalars ($\tau _pML$ complex scalars for the pilot signals and $(\tau _c-\tau _p)ML$ complex scalars for the received data signals). Besides, the full correlation matrices $\{\mathbf{R}_{mk}\}$ are available at the CPU, which contains ${{MKL^2N^2}/{2}}$ complex scalars for each realization of the AP/UE locations/statistics\footnote{Note that the channel statistics remain constant for each realization of the AP/UE locations and each realization of the AP/UE locations contains $N_r$ channel realizations (coherence blocks).}. Moreover, the CPU transmits optimal precoding matrices to all UEs, which are described by $KN^2$ complex scalars per coherence block. In summary, for the FCP scheme with the I-WMMSE precoding implemented, total $\tau _cMLN_r+{{MKL^2N^2}/{2}}+KN^2N_r$ complex scalars are transmitted via fronthaul links for each realization of the AP/UE locations. For comparison, when the FCP scheme without the I-WMMSE precoding is implemented, all APs should also transmit $\tau _cML$ complex scalars for the received signals to the CPU in each coherence block and ${{MKL^2N^2}/{2}}$ complex scalars for $\{\mathbf{R}_{mk}\}$ to the CPU for each realization of the AP/UE locations. So for the CFP scheme without the I-WMMSE precoding, total $\tau _cMLN_r+{{MKL^2N^2}/{2}}$ complex scalars are transmitted via fronthaul links for each realization of the AP/UE locations.

As for the LSFD scheme with the I-WMMSE precoding, all APs transmit their local data estimates $\mathbf{\tilde{x}}_{mk}$, described by $(\tau _c-\tau _p)MKN$ complex scalars, to the CPU per coherence block. Besides, $\mathbb{E}\{ \mathbf{G}_{kk} \}\in \mathbb{C}^{MN\times N}$, described by $MKN^2$ complex scalars for each realization of the AP/UE locations, are also required at the CPU. As for $\mathbb{E} \{ \mathbf{G}_{kl}\mathbf{\bar{F}}_{l,\mathrm{u},(2)}\mathbf{G}_{kl}^{H} \}\in \mathbb{C}^{MN\times MN}$, following the formulation method investigated in Appendix~\ref{F_Closed_Form}, the optimization of $\mathbf{\bar{F}}_{l,\mathrm{u},(2)}$ requires the knowledge of $\left\{ \mathbb{E} \left[ \mathbf{v}_{ml,p}^{H}\mathbf{h}_{mk,n}\mathbf{h}_{m^{\prime}k,i}^{H}\mathbf{v}_{m^{\prime}l,p^{\prime}} \right] \right\}$, described by $M^2K^2N^4/2$ complex scalars for each realization of the AP/UE locations, where $\mathbf{v}_{ml,p}$ denotes the $p$-th column of $\mathbf{V}_{ml}$. Moreover, the CPU requires optimal precoding matrices $\mathbf{F}_{k,\mathrm{u},(2)}$ feedback to all APs and UEs only for each realization of the AP/UE locations, which are $KN^2$ complex scalars. As for the LSFD scheme without the I-WMMSE precoding, local data estimates $\mathbf{\tilde{x}}_{mk}$, described by $(\tau _c-\tau _p)MKN$ complex scalars per coherence block, $\mathbb{E}\{ \mathbf{G}_{kk} \}$, described by $MKN^2$ complex scalars for each realization of the AP/UE locations, and $\mathbb{E} \{ \mathbf{G}_{kl}\mathbf{G}_{kl}^{H} \}\in \mathbb{C}^{MN\times MN}$, described by $M^2K^2N^2/2$ complex scalars for each realization of the AP/UE locations, are required. That is total $( \tau _c-\tau _p ) MKNN_r+MKN^2+M^2K^2N^2/2$ complex scalars transmitted via fronthaul links for each realization of the AP/UE locations.

\subsubsection{Practical Implementation}
Note that the basic motivation of the investigated I-WMMSE precoding schemes is to achieve as good the sum uplink SE performance as possible so we ignore some practical issues, which are vital for the realistic implementation of the investigated precoding schemes. When the precoding schemes are implemented in practice, these realistic issues should be considered.

$\bullet$  Capacity-constrained fronthaul network

As discussed above, the I-WMMSE precoding require more fronthaul requirements than the case without the I-WMMSE precoding. It is quite vital to consider a more practical capacity-constrained fronthaul network \cite{8756286}. Moreover, the wireless fronthaul \cite{8283646}, which is more flexible than the conventional wire fronthaul, would also be regarded as a promising solution to boost the practical implementation of the I-WMMSE precoding.

$\bullet$  Scalability aspects with dynamic cooperation clusters

When the precoding schemes are implemented in practice, a more realistic network architecture with multiple CPUs and dynamic cooperation clusters should be advocated, where each UE is only served by a cluster of APs (that a is user-centric cluster) and the APs are grouped into cell-centric clusters as shown in Fig. \ref{System_Model}. Note a user-centric cluster might consist of APs connecting with different CPUs. Based on the signal processing schemes in \cite{9064545,8761828}, the analytical framework in this paper can be implemented in a scalable paradigm where the fronthaul requirements and computational complexity can be relieved with an anticipated modest performance loss compared with canonical architecture. The I-WMMSE precoding design with these two practical aspects is left in future work. To bring valuable technical insights for the study of I-WMMSE precoding schemes with the DCC strategy and the capacity-constrained fronthaul link, we provide two tutorials for the FCP and LSFD in Fig. \ref{Practical_Tutorial} based on \cite{9064545,9174860,8756286}.

\begin{figure}[ht]
\centering
\includegraphics[scale=0.47]{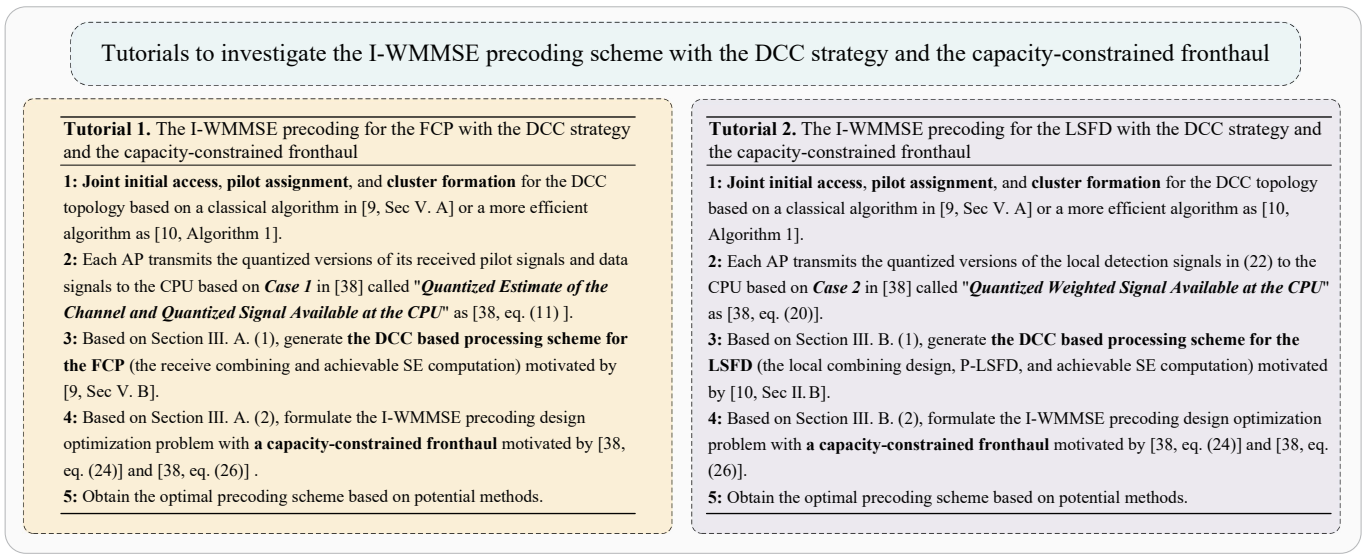}
\caption{Two tutorials to investigate the I-WMMSE precoding schemes with the DCC strategy and the capacity-constrained fronthaul.
\label{Practical_Tutorial}}
\vspace{-0.4cm}
\end{figure}

\begin{table*}[tp]
  \centering
  \fontsize{9}{10}\selectfont
  \caption{Comparison of two precoding schemes in this paper. The number of complex scalars is computed for each realization of the AP/UE locations. The sum SE improvement is computed with $M=20$, $K=10$, $L=1$ and $N=4$.}
  \label{Paper_comparison}
    \begin{tabular}{  !{\vrule width1.2pt} m{4cm}<{\centering}  !{\vrule width1.2pt} m{4cm}<{\centering} !{\vrule width1.2pt} m{6.5cm}<{\centering} !{\vrule width1.2pt} }

    \Xhline{1.2pt}
         \cellcolor{gray!30} \bf  &  \cellcolor{gray!30} \bf FCP  & \cellcolor{gray!30} \bf LSFD\cr
    \Xhline{1.2pt}

         CSI  & \makecell[c]{Instantaneous CSI}  & \makecell[c]{Statistical CSI} \cr\hline
        \Xhline{0.8pt}
         Detection scheme  & \makecell[c]{MMSE combining}  & \makecell[c]{L-MMSE/MR combining + Optimal LSFD scheme} \cr\hline
        \Xhline{0.8pt}
        Number of complex scalars sent from APs to the CPU with I-WMMSE precoding & \makecell[c]{ $\tau _cMLN_r+{{MKL^2N^2}/{2}}$}  & \makecell[c]{$( \tau _c-\tau _p ) MKNN_r+MKN^2+M^2K^2N^4/2$} \cr\hline
        \Xhline{0.8pt}
        Number of complex scalars sent from APs to the CPU without I-WMMSE precoding & \makecell[c]{ $\tau _cMLN_r+{{MKL^2N^2}/{2}}$}  & \makecell[c]{$( \tau _c-\tau _p ) MKNN_r+MKN^2+M^2K^2N^2/2$} \cr\hline
        \Xhline{0.8pt}
         Number of complex scalars feedback sent from the CPU  & \makecell[c]{$KN^2N_r$}  & \makecell[c]{$KN^2$} \cr\hline
        \Xhline{0.8pt}

        Per-iteration computational complexity  & \makecell[c]{$\mathcal{O} \left( M^3K^2N^5N_r \right) $}  & \makecell[c]{L-MMSE: $\mathcal{O} \left( M^2K^2N^3N_r \right) $\\MR (Monte-Carlo): $\mathcal{O} \left( M^2K^2N^3N_r+M^3KN^3 \right) $ \\MR (Analytical): $\mathcal{O} \left( M^3K^2N^5 \right) $} \cr\hline
        \Xhline{0.8pt}

        Sum SE improvement  & \makecell[c]{$28.93\%$}  & \makecell[c]{L-MMSE: $46.74\%$ \\MR: $15.13\%$} \cr\hline

    \Xhline{1.2pt}
    \end{tabular}
  \vspace{0cm}
\end{table*}

\vspace{-0.6cm}
\subsection{Complexity Analysis}\label{comAna}
In this subsection, we analyze the computational complexity of two precoding schemes investigated. Since the bisection step for $\lambda _{k,\{(1),(2)\}}$ generally takes few iterations compared with other steps, we ignore bisection steps for $\lambda _{k,\{(1),(2)\}}$ in the complexity analysis. For the fully centralized processing scheme and each realization of the AP/UE locations, the per-iteration complexity of iterative optimization is $\mathcal{O} \left( M^3K^2N^5N_r \right) $. For the LSFD scheme and each realization of the AP/UE locations, the per-iteration complexity of iterative optimization based on L-MMSE combining with the Monte-Carlo method, MR combining with the Monte-Carlo method and MR combining with the closed-form expressions are $\mathcal{O} \left( M^2K^2N^3N_r \right) $, $\mathcal{O} \left( M^2K^2N^3N_r+M^3KN^3 \right) $ and $\mathcal{O} \left( M^3K^2N^5 \right) $, respectively. To further reduce the computation complexity, it's quite necessary to apply the asymptotic analysis method \cite{6415388,6172680} to compute the terms, which cannot be computed in closed-form, in approximation results.

\section{Numerical Results}\label{Numerical_Results}

%

%

In this paper, a CF mMIMO system is investigated, where all APs and UEs are uniformly distributed in a $1\times1\,\text{km}^2$ area with a wrap-around scheme \cite{zhang2021improving}. The pathloss and shadow fading are modeled similarly as \cite{04962}. In practice, $\mathbf{U}_{mk,\mathrm{r}}$, $\mathbf{U}_{mk,\mathrm{t}}$ and $\mathbf{\Omega }_{mk}$ are estimated through measurements \cite{1576533}. However, we generate them randomly in this paper, where the coupling matrix $\mathbf{\Omega }_{mk}$ consists of one strong transmit eigendirection capturing dominant power \cite{1459054}\footnote{In this paper, we choose one eigendirection capturing dominant channel power (randomly accounting for $80\% \sim 95\%$ of the total channel power) and other eigendirections contain the remaining power.}. Besides, we have $\mathbf{F}_{k,\mathrm{p}}=\mathbf{F}_{k,\mathrm{u},\{(1),(2)\}}^{\left( 0 \right)}=\sqrt{\frac{p_k}{N}}\mathbf{I}_N$. As for Algorithm \ref{algo:iterative1} and Algorithm \ref{algo:iterative}, balancing the convergence and accuracy, we assume that $I_{(1),\max}=I_{(2),\max}=20$, $\mathrm{\varepsilon}_{(1)}=\mathrm{\varepsilon}_{(2)}=5\times 10^{-4}$, and weights for all UEs are equal ($\mu _{k,(1)}=\mu _{k,(2)}=1$) without losing generality, respectively. Moreover, we consider communication with $20\,\text{MHz}$ bandwidth and $\sigma ^2=-94\,\text{dBm}$ noise power. All UEs transmit with $200\,\text{mW}$ power constraint. Each coherence block contains $\tau _c=200$ channel uses and $\tau _p=KN/2$. Besides, a pilot assignment approach similar as that in \cite{04962} is investigated.

\begin{figure}[t]
\begin{minipage}[t]{0.48\linewidth}
\centering
\includegraphics[scale=0.5]{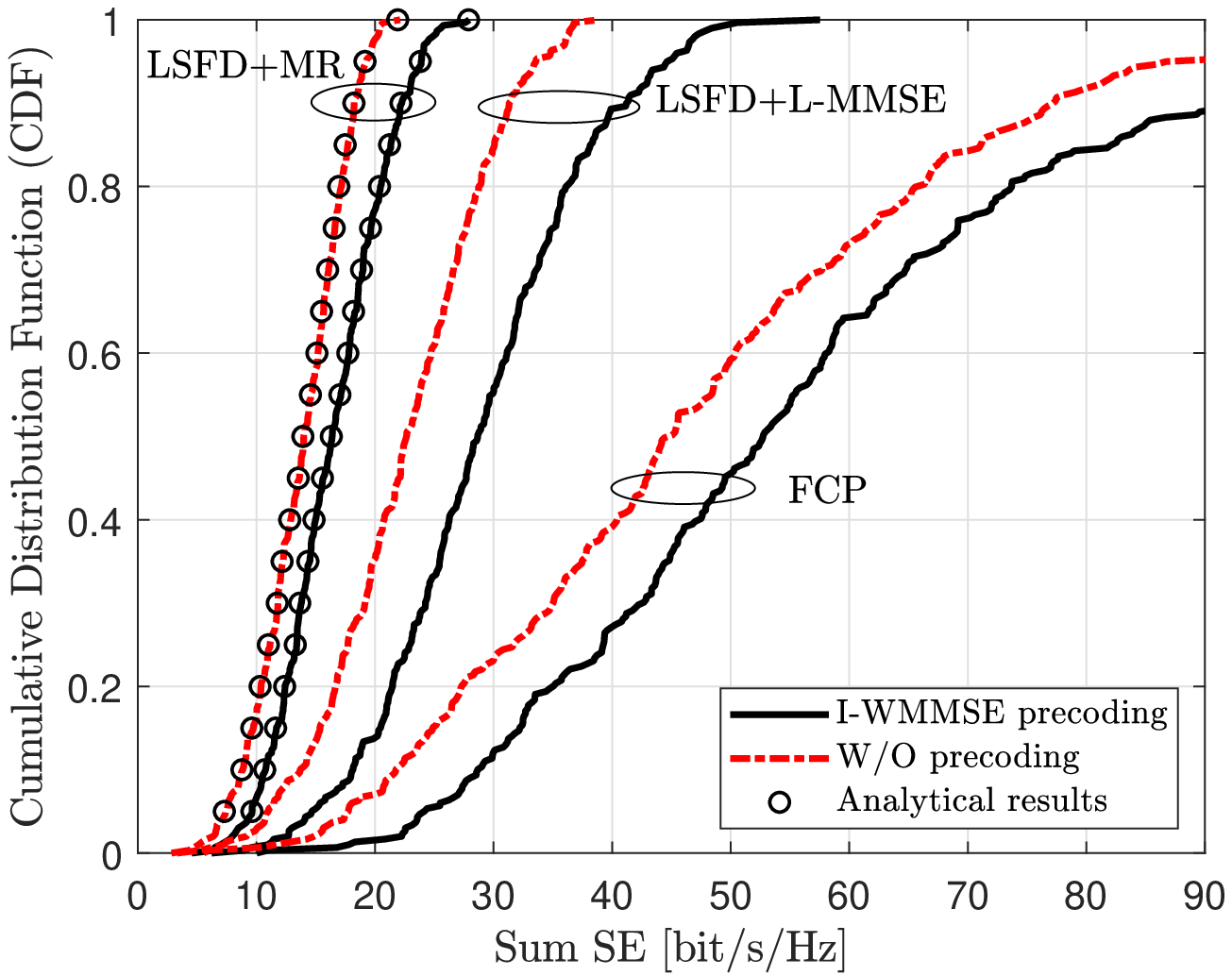}
\caption{CDF of the sum SE over different processing schemes and precoding schemes with $M=20$, $K=10$, $L=2$, and $N=4$.
\label{fig1:SE_CDF}}
\end{minipage}
\hfill
\begin{minipage}[t]{0.48\linewidth}
\centering
\includegraphics[scale=0.5]{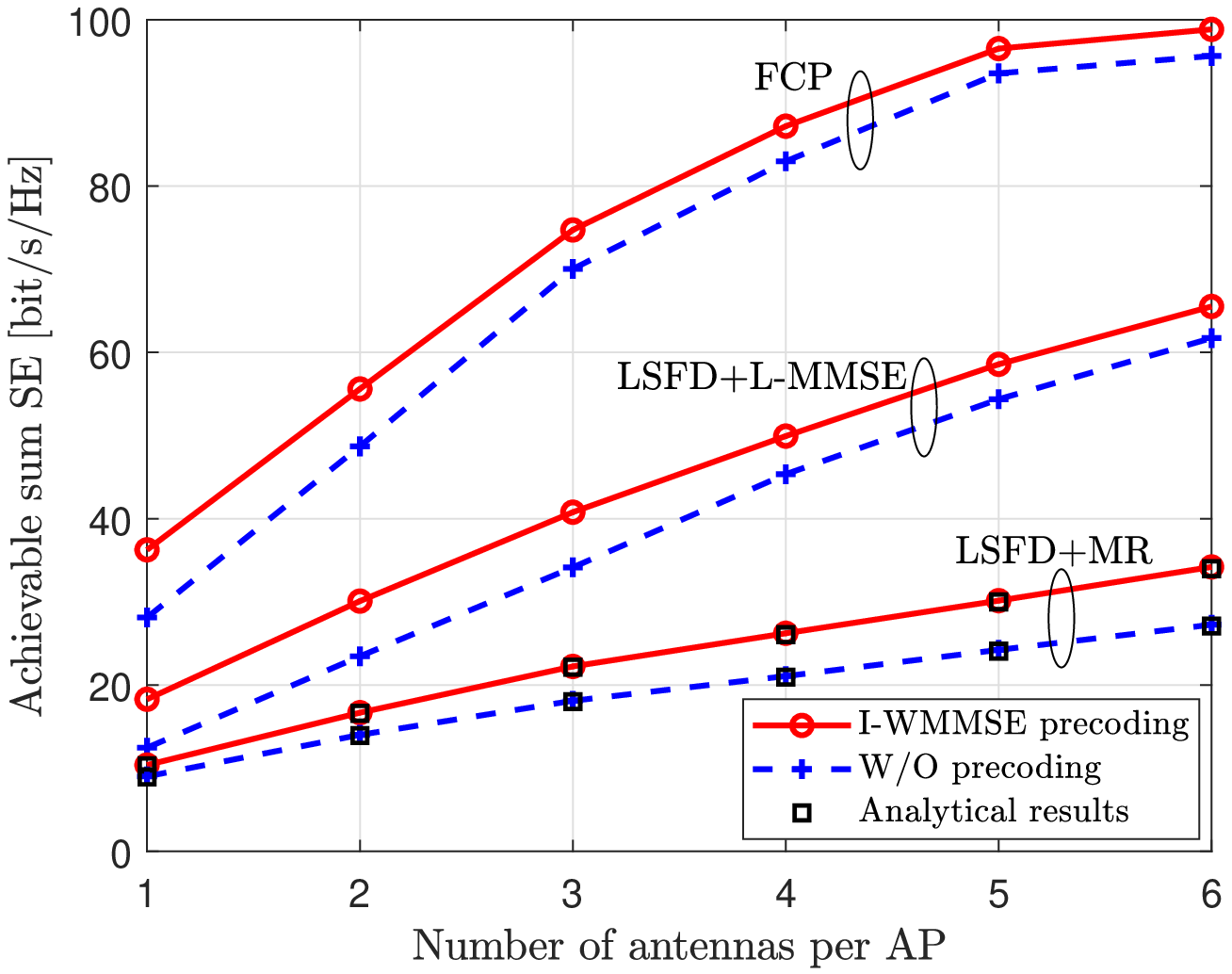}
\caption{Sum SE against the number antennas per AP $L$ over different processing schemes and precoding schemes with $M=20$, $K=10$, and $N=4$.
\label{fig1:SE_L}}
\end{minipage}
\end{figure}

\begin{figure}[t]
\begin{minipage}[t]{0.48\linewidth}
\centering
\includegraphics[scale=0.5]{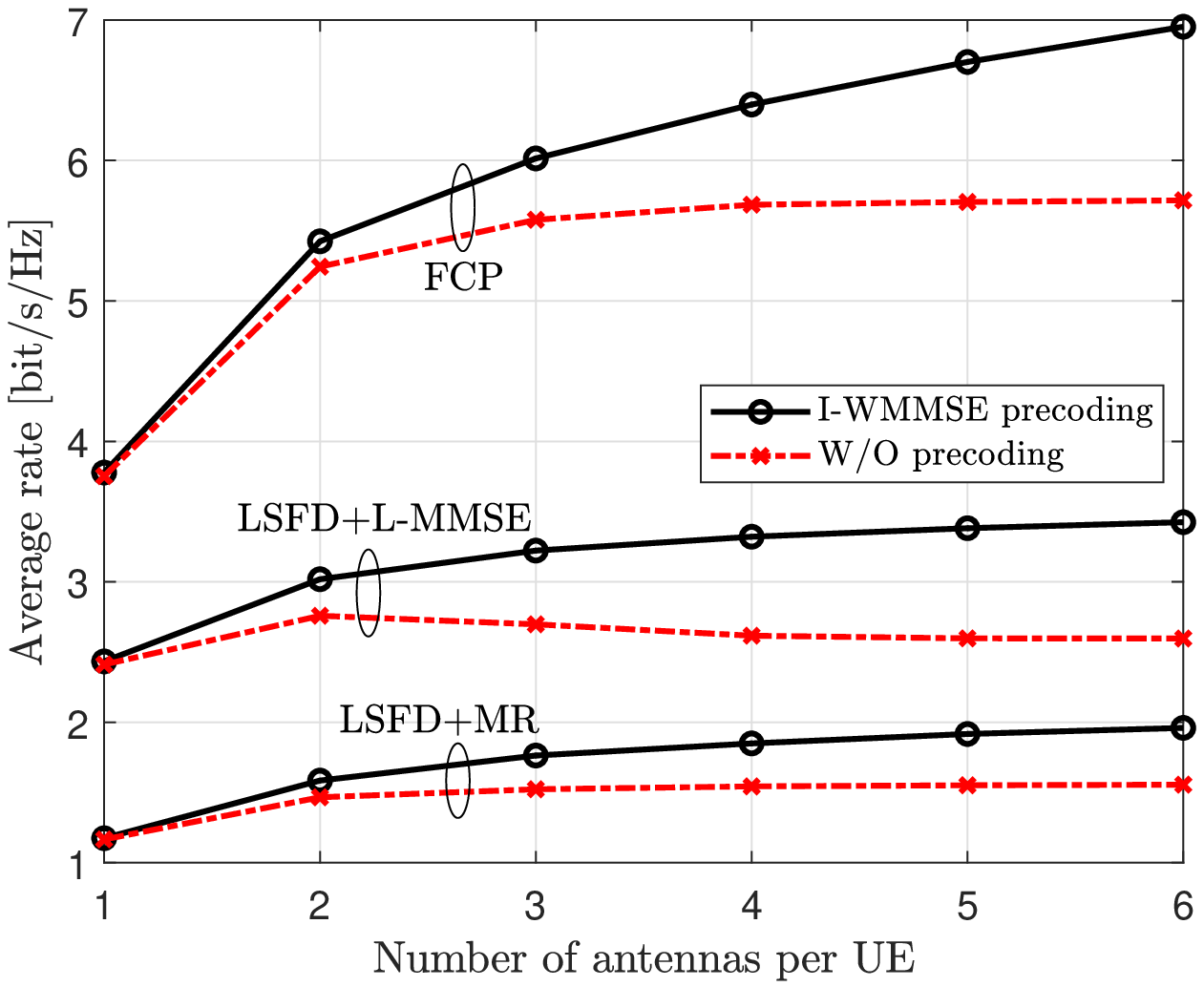}
\caption{Average rate against the number of antennas per UE $N$ over different processing schemes and precoding schemes with $M=20$, $K=10$, and $L=2$.\label{Rate_N}}
\end{minipage}
\hfill
\begin{minipage}[t]{0.48\linewidth}
\centering
\includegraphics[scale=0.5]{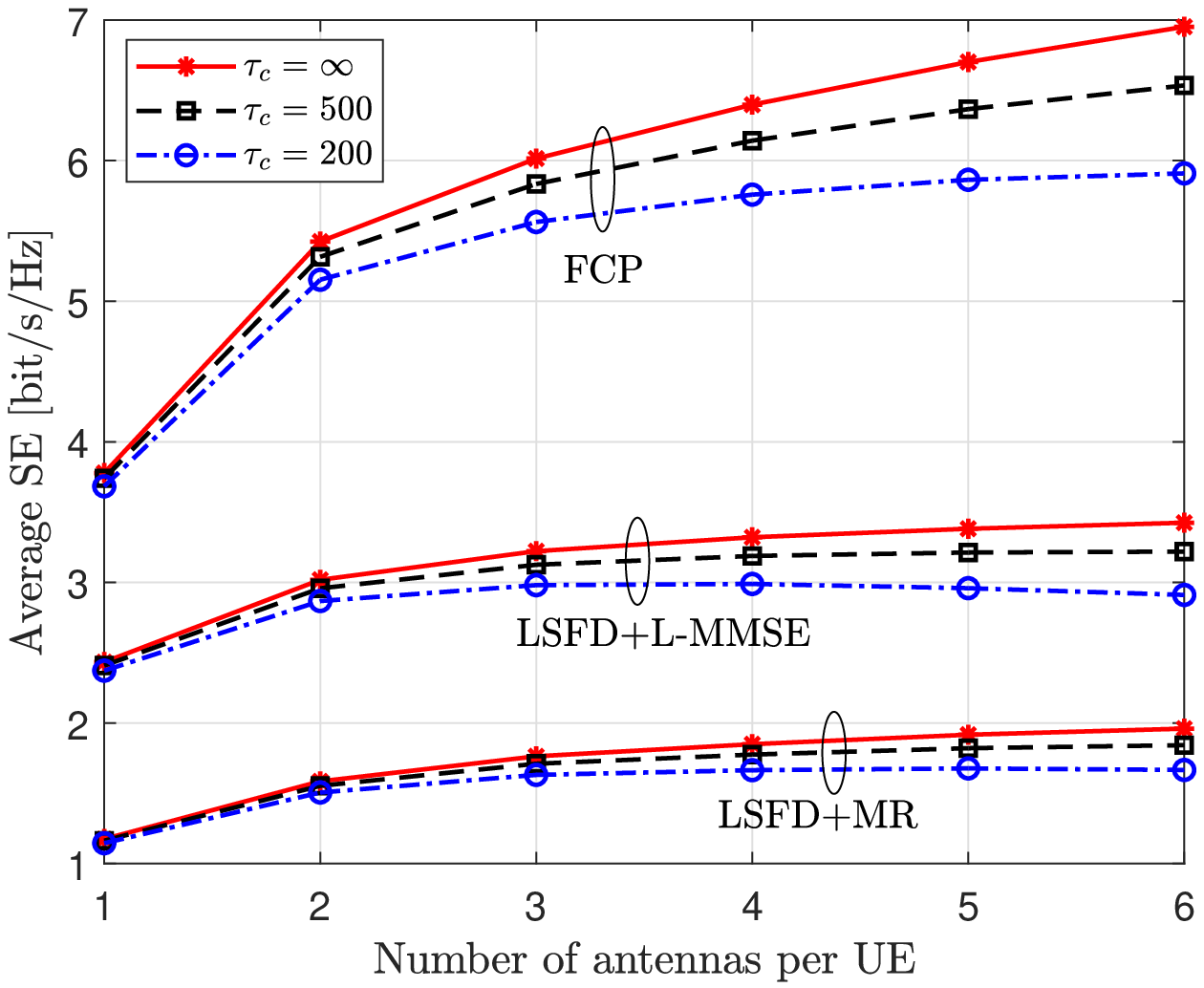}
\caption{Average SE with I-WMMSE precoding schemes against the number of antennas per UE $N$ over different $\tau_c$ with $M=20$, $K=10$, and $L=2$.\label{SE_tau_c}}
\end{minipage}
\vspace{-0.4cm}
\end{figure}

\begin{figure}[t]
\begin{minipage}[t]{0.48\linewidth}
\centering
\includegraphics[scale=0.5]{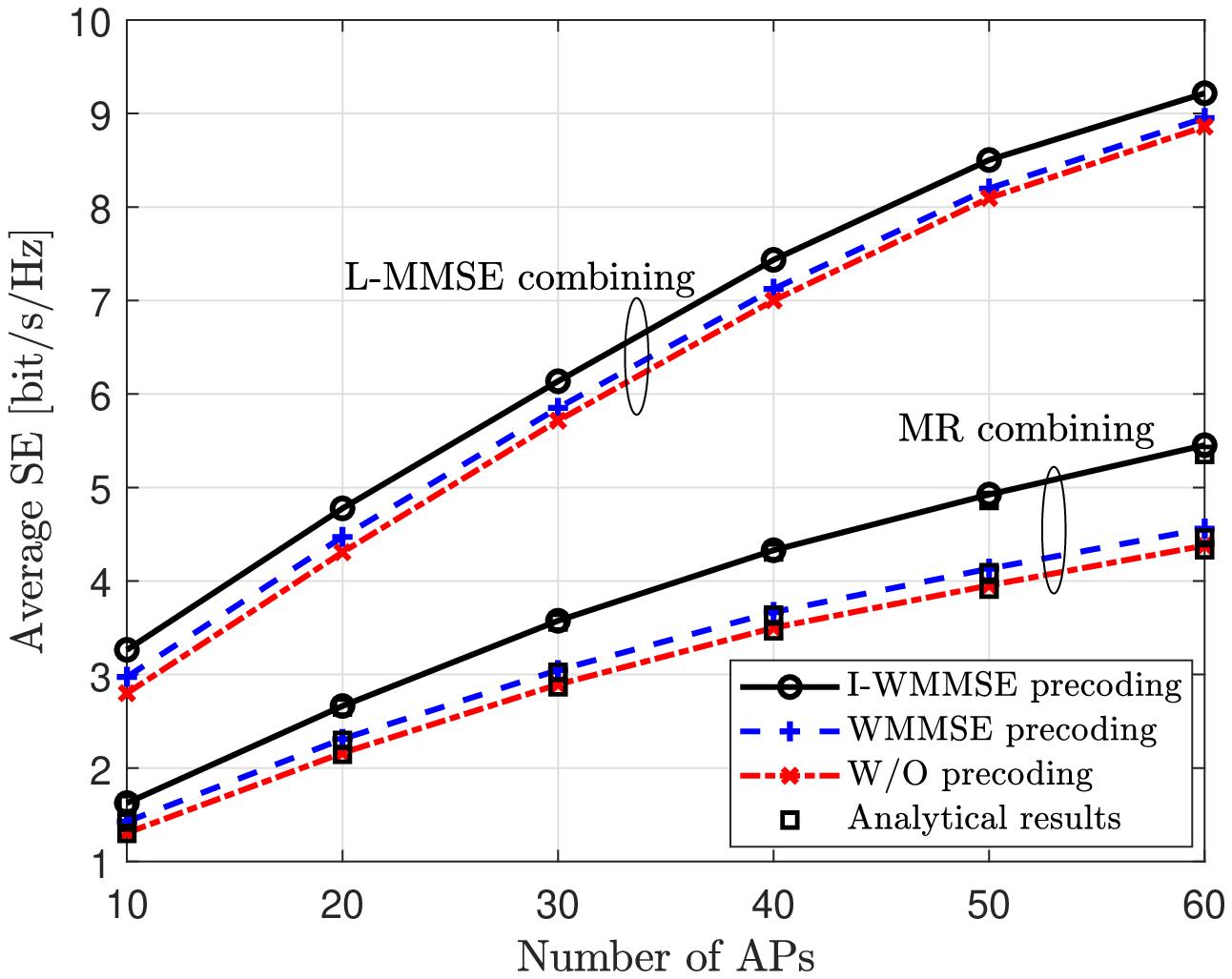}
\caption{Average SE against the number of APs $M$ for the LSFD scheme with $K=10$, $L=4$, and $N=4$.
\label{SE_M}}
\end{minipage}
\hfill
\begin{minipage}[t]{0.48\linewidth}
\centering
\includegraphics[scale=0.5]{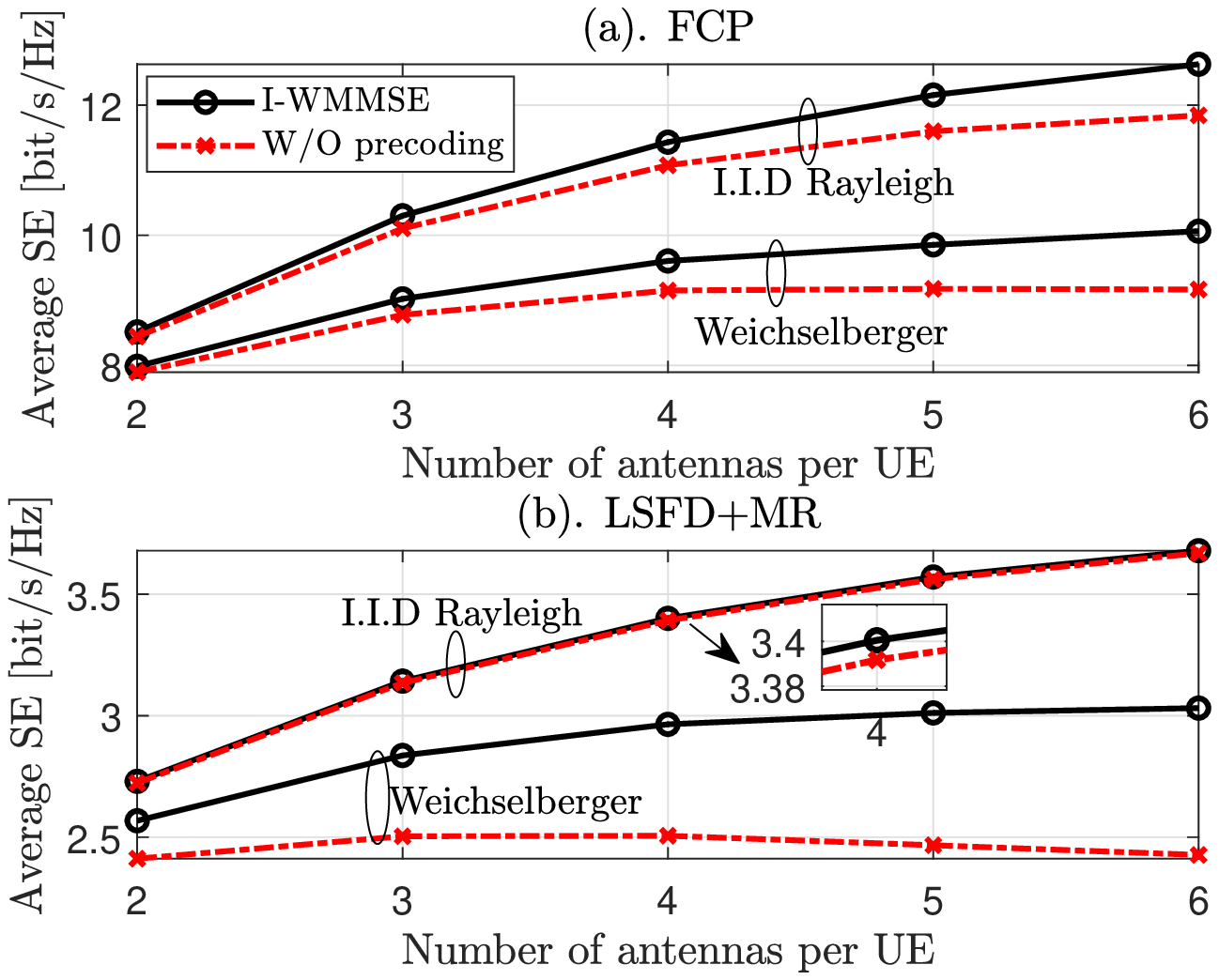}
\caption{Average SE against the number of antennas per UE $N$ for different channel models with $M=40$, $K=8$, and $L=2$.\label{SE_channel}}
\end{minipage}
\vspace{-0.4cm}
\end{figure}

\begin{figure}[t]\centering
\subfigure[FCP]{
\begin{minipage}{8cm}\centering
\includegraphics[scale=0.5]{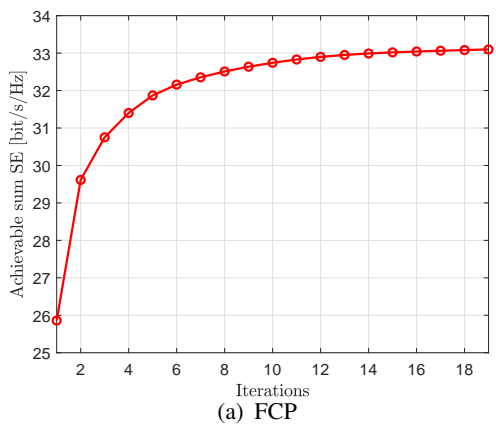}
\end{minipage}}
\subfigure[LSFD]{
\begin{minipage}{8cm}\centering
\includegraphics[scale=0.5]{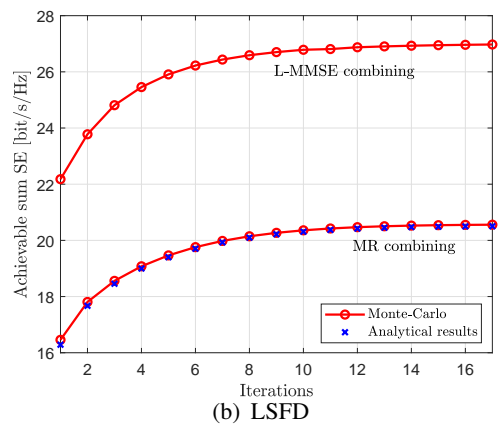}
\end{minipage}}
\caption{Convergence examples of the I-WMMSE algorithm for the FCP and LSFD with $M=20$, $K=10$, $L=2$, and $N=4$.
\label{Convergence}}
\vspace{-0.3cm}
\end{figure}

\begin{figure}[t]
\centering
\includegraphics[scale=0.5]{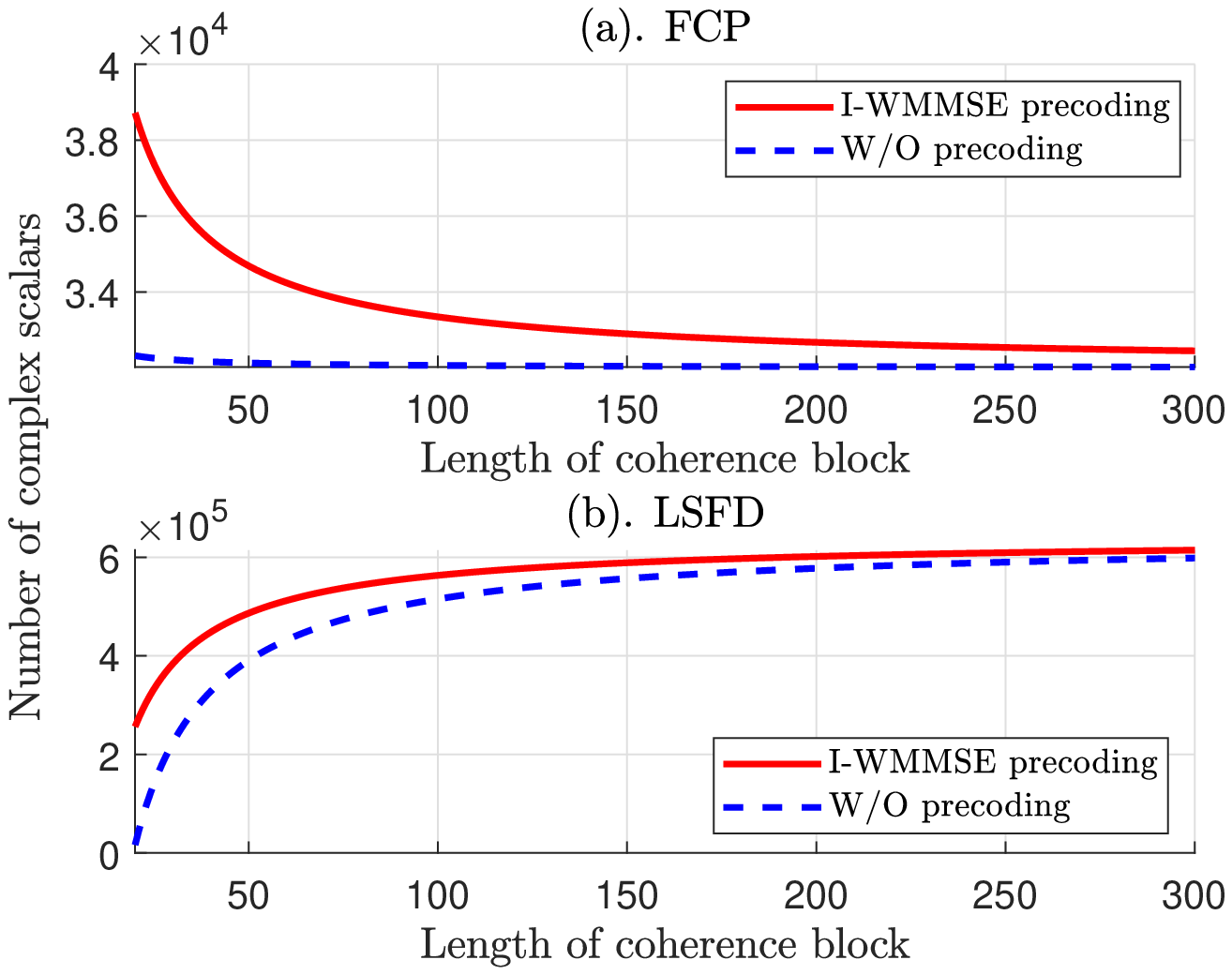}
\caption{Total number of complex scalars sent via the fronthaul per channel use for each realization of the AP/UE locations with $M=20$, $K=10$, $L=2$, and $N=4$.
\label{Complex_scalars}}
\vspace{-0.7cm}
\end{figure}

Figure \ref{fig1:SE_CDF} shows the cumulative distribution function (CDF) of the achievable sum SE over different realizations of the AP/UE locations for two processing schemes investigated (we shortly call ``fully centralized processing" as ``FCP" in the following) over ``I-WMMSE precoding" or ``w/o precoding"\footnote{The ``w/o precoding" scenario denotes that identity precoding matrices $\mathbf{F}_{k,\mathrm{u},\{(1),(2)\}}=\sqrt{\frac{p_k}{N}}\mathbf{I}_N$ are implemented without optimization.}. We notice that the FCP scheme undoubtedly achieves higher SE than that of the LSFD scheme since the FCP with MMSE combining is a competitive scheme in CF mMIMO \cite{[162]}. More importantly, the proposed I-WMMSE schemes are efficient to improve the respective achievable sum SE performance, e.g., $12.78\%$, $19.54\%$ and $28.13\%$ sum SE improvement for the FCP, the LSFD with MR combining and the LSFD with L-MMSE combining, respectively. Besides, for the LSFD scheme with MR combining, markers ``$\circ$"  generated by analytical results overlap with the curves generated by simulations, respectively, validating our derived closed-form expressions.

Figure \ref{fig1:SE_L} shows the achievable sum SE as a function of the number of antennas per AP with two processing schemes investigated and different precoding schemes\footnote{Note that the achievable sum SE investigated is the average sum SE value taken over many AP/UE locations.}. We notice that, for the FCP or LSFD with (L-)MMSE combining, the performance gap between the ``I-WMMSE" and ``w/o precoding" becomes smaller with the increase of $L$, which implies that (L-)MMSE combining can use all antennas on each AP to suppress interference and achieve excellent SE performance even without any precoding scheme. For instance, the performance gap between the ``I-WMMSE" and ``w/o precoding" for the LSFD with L-MMSE combining is $46.74\%$ and $6.17\%$ over $L=1$ and $L=6$, respectively. Meanwhile, for the LSFD with MR combining, the performance gap between the ``I-WMMSE" and ``w/o precoding" becomes large with the increase of L, e.g. $15.13\%$ and $25.48\%$ for $L=1$ and $L=6$, respectively. Besides, for the LSFD scheme with MR combining, markers ``$\Box $"  generated by analytical results overlap with the curves generated by simulations, respectively, validating our derived closed-form expressions.

To further show the advantage of the proposed I-WMMSE precoding schemes, Fig. \ref{Rate_N} shows the average rate\footnote{Note that one main reason for the phenomenon that additional UE antennas may give rise to the SE degradation is that increasing $N$ will increase the channel estimation overhead and reduce the pre-log factor ``$\left( \tau _c-\tau _p \right)/ \tau_c$" in all SE expressions \cite{194,04962}. So we investigate ``the average rate" in Fig. \ref{Rate_N}, ignoring the effect of ``$\left( \tau _c-\tau _p \right)/ \tau_c$".} as a function of the number of antennas per UE. We find that the average rates for all schemes with I-WMMSE precoding schemes grow with $N$ and the average rates for the case without UL precoding may also suffer the degradation with the increase of $N$. The implementation of the I-WMMSE precoding schemes undoubtedly makes UEs benefit from multiple antennas and achieve excellent rate performance. Moreover, we observe that the I-WMMSE precoding schemes perform more efficiently with a larger number of UE antennas. For instance, the average rate improvements achieved by the I-WMMSE precoding for the LSFD with L-MMSE combining are $31.91\%$ and $9.43\%$ for $N=6$ and $N=2$, respectively. However, the average SE (with scaling factor $(\tau_c-\tau_p)/\tau_c$) with I-WMMSE precoding implemented may also degrade with the increase of $N$ as the Fig. 2 in \cite{2022arXiv220111299W} since, with the increase of $N$, the prerequisite of ``mutually orthogonal pilot matrices" still requires huge channel uses for the pilot transmission and the inter-user interference also increases. So the design of non-orthogonal pilot matrices and per-antenna power control scheme are quite necessary, which are regarded as promising ways to reduce the cost of pilot transmission and further improve the SE performance \cite{9169906}.

Figure \ref{SE_tau_c} discusses the average SE with I-WMMSE precoding schemes against $N$ over different $\tau_c$. Note that Fig. \ref{Rate_N} can be viewed as a special case in Fig. \ref{SE_tau_c} with the coherence block with infinite length $\tau _c=\infty$. We observe that the average SE with I-WMMSE precoding schemes increases with $N$ over $\tau _c=500$ or $\infty$, which means the SE performance can benefit from having additional UE antennas when the coherence block resource is abundant.

Figure \ref{SE_M} investigates the average SE as a function of $M$ for the LSFD scheme over different precoding schemes\footnote{The ``WMMSE precoding" denotes the precoding schemes generated by the I-WMMSE algorithm with only single iteration.}. For MR combining, markers ``$\Box $"  generated by analytical results overlap with the curves generated by simulations, respectively, validating our derived closed-form expressions again. Besides, the I-WMMSE algorithm is more efficient to improve the SE performance for MR combining than that of L-MMSE combining for the scenario over large $L$ and $M$, e.g., $4.03\%$ and $24.21\%$ SE improvement for L-MMSE combining and MR combining with $M=60$, respectively, implying that the L-MMSE combining based on large $L$ and $M$ can achieve excellent SE performance even without any precoding scheme and the proposed I-WMMSE precoding scheme is handy to mitigate the weakness of MR combining\footnote{MR combining is a simple combining scheme but cannot efficiently suppress the interference.}.

Figure \ref{SE_channel} considers the average SE as a function of $N$ over the i.i.d. and the Weichselberger Rayleigh fading channel. As observed, the proposed I-WMMSE precoding schemes are more efficient over the Weichselberger Rayleigh fading channel. For instance, $24.89\%$ and $9.77\%$ average SE improvement can be achieved when $N=6$ over the ``Weichselberger" scenario for the LSFD scheme with MR combining and the FCP scheme, respectively, but only $0.29\%$ and $6.63\%$ average SE improvement can be achieved for ``I.I.D. Rayleigh channel". Moreover, compared with Fig. \ref{Rate_N}, we notice that the I-WMMSE precoding scheme for the FCP scheme is more efficient in the highly loaded system (the scenario in Fig. \ref{Rate_N}) where the number of total AP-antennas is comparable with the number of total UE-antennas.

Figure \ref{Convergence} illustrates the convergence behavior of the I-WMMSE algorithms for the FCP scheme and the LSFD scheme with L-MMSE/MR combining. Note the convergence example in Fig. \ref{Convergence} (a)  for the FCP is given by a particular channel realization and the convergence example for the LSFD in Fig. \ref{Convergence} (b) is given by a particular realization of the AP/UE locations. Note that the algorithms investigated can be guaranteed to converge and are efficient to achieve excellent sum SE performance. Besides, Fig. \ref{Convergence} (b) for the LSFD scheme with MR combining validates our derived closed-form expressions in Algorithm~\ref{algo:iterative}.

Figure \ref{Complex_scalars} investigates the total number of complex scalars sent via the fronthaul per channel use against $\tau_c$ for each realization of the AP/UE locations. As observed, total number of complex scalars per channel use for the FCP/LSFD scheme becomes smaller/larger, which can also be easily found from Table~\ref{Paper_comparison}. Besides, the LSFD scheme requires more fronthaul signaling than the FCP scheme since APs under the LSFD scheme need to transmit all received data signals to the CPU, which requires a huge fronthaul load. More importantly, with the increase of $\tau_c$, the gap between ``I-WMMSE precoding" and ``W/O precoding" becomes smaller for either the FCP scheme or the LSFD scheme. Considering the SE performance improvement of the I-WMMSE precoding, additional fronthaul loads can be acceptable, especially when the coherence resource is abundant. Although the computational complexity of Algorithm~\ref{algo:iterative1} for the FCP scheme is much higher than that of Algorithm~\ref{algo:iterative} for the LSFD scheme, the FCP scheme needs much less fronthaul signaling than that of the LSFD scheme and can achieve better SE performance. So two processing schemes and their respective precoding schemes can be chosen based on different requirements.



\section{Conclusion}\label{sec:conclusion}
We consider a CF mMIMO system with both APs and UEs equipped with multiple antennas over the Weichselberger Rayleigh fading channel. The FCP scheme and LSFD scheme are implemented. To further improve the sum SE performance, efficient UL precoding schemes based on iteratively WMMSE algorithms are investigated to maximize weighted sum SE for the two processing schemes. Note that we compute achievable SE expressions and optimal precoding schemes in novel closed-form for the LSFD scheme with MR combining. Numerical results show that the investigated I-WMMSE precoding schemes are efficient to achieve excellent sum SE performance. More importantly, it can be seen that the proposed I-WMMSE precoding schemes are more efficient with a larger number of UE antennas, which means the I-WMMSE precoding schemes can achieve excellent performance even with a large number of UE antennas. The derived results undoubtedly provides vital insights for the practical implementation of multi-antenna UEs in CF mMIMO systems. In future work, we will investigate the design of UL precoding scheme for the phase of pilot transmission and consider the practical implementation of the investigated I-WMMSE precoding schemes with capacity-constrained fronthaul network and dynamic cooperation clusters. Moreover, the non-orthogonal pilot matrix design will also be considered to further improve the performance for the CF mMIMO system with multi-antenna UEs. Last but not least, the UL precoding performance over a more practical Rician fading channel with phase-shifts will also be analyzed.
\vspace{-0.6cm}

\begin{appendices}

\section{Some useful Lemmas}
\begin{lemm}\label{Lemma1}
Let $\mathbf{X}\in \mathbb{C} ^{M\times N}$ be a random matrix and $\mathbf{Y}$ is a deterministic $M\times M$ matrix. The $(n,i)$-th element of $\mathbb{E} \left\{ \mathbf{X}^H\mathbf{YX} \right\}$ is $\mathrm{tr}\left( \mathbf{Y}\cdot \mathbb{E} \left\{ \mathbf{x}_i\mathbf{x}_{n}^{H} \right\} \right)$ where $\mathbf{x}_i$ and $\mathbf{x}_n$ are the $i$-th and $n$-th column of $\mathbf{X}$.
\end{lemm}
\begin{lemm}\label{Lemma2}
For matrices $\mathbf{A}\in \mathbb{C} ^{N_1\times N_1}$, $\mathbf{B}\in \mathbb{C} ^{N_1\times N_2}$, $\mathbf{C}\in \mathbb{C} ^{N_2\times N_2}$, and $\mathbf{D}\in \mathbb{C} ^{N_2\times N_1}$, we have
$
\left( \mathbf{A}+\mathbf{BCD} \right) ^{-1}=\mathbf{A}^{-1}-\mathbf{A}^{-1}\mathbf{B}\left( \mathbf{DA}^{-1}\mathbf{B}+\mathbf{C}^{-1} \right) ^{-1}\mathbf{D}\mathbf{A}^{-1},
$
which is a well-known matrix inversion lemma \cite[Lemma B.3]{8187178}.
\end{lemm}
\vspace{-0.9cm}
\section{Proof Corollary~\ref{Corollary_Optimal_LSFD}}\label{Appendix1}
Since the CPU is only aware of channel statistics, we need to treat $\mathbb{E} \{ \mathbf{G}_{kk}\} \mathbf{F}_{k,\mathrm{u}}$ as the true deterministic channel and rewrite ${\mathbf{\tilde{x}}_k}$ in \eqref{Data_Final_LSFD_Compact} as
$
\mathbf{\tilde{x}}_k=\mathbb{E} \left\{ \mathbf{G}_{kk} \right\} \mathbf{F}_{k,\mathrm{u}}\mathbf{x}_k+\underset{\mathbf{v}}{\underbrace{\left( \mathbf{G}_{kk}\mathbf{F}_{k,\mathrm{u}}-\mathbb{E} \left\{ \mathbf{G}_{kk} \right\} \mathbf{F}_{k,\mathrm{u}} \right) \mathbf{x}_k+\sum_{l=1,l\ne k}^K{\mathbf{G}_{kl}\mathbf{F}_{l,\mathrm{u}}\mathbf{x}_l}+\mathbf{n}_{k}^{\prime}}}
$
where $\mathbf{v}$ is a complex circular symmetric noise with an invertible covariance matrix
$
\mathbf{\Xi }_k=\mathbb{E} \{ \mathbf{vv}^H |\mathbf{\Theta }\} =\sum_{l=1}^K{\mathbb{E} \{ \mathbf{G}_{kl}\mathbf{F}_{k,\mathrm{u}}\mathbf{F}_{k,\mathrm{u}}^{H}\mathbf{G}_{kl}^{H} \}}-\mathbb{E} \{ \mathbf{G}_{kk} \} \mathbf{F}_{k,\mathrm{u}}\mathbf{F}_{k,\mathrm{u}}^{H}\mathbb{E} \{ \mathbf{G}_{kk}^{H} \} +\sigma ^2\mathbf{S}_k.
$
Firstly, we whiten the noise as
$
\mathbf{\Xi }_{k}^{-\frac{1}{2}}\mathbf{\hat{x}}_k=\mathbf{\Xi }_{k}^{-\frac{1}{2}}\mathbb{E} \left\{ \mathbf{G}_{kk} \right\} \mathbf{F}_{k,\mathrm{u}}\mathbf{x}_k+\mathbf{\tilde{v}},
$
where $\mathbf{\tilde{v}}\triangleq \mathbf{\Xi }_{k}^{-\frac{1}{2}}\mathbf{v}$ becomes white.
Next, we project $\mathbf{\Xi }_{k}^{-\frac{1}{2}}\mathbf{\hat{x}}_k$ in the direction of $\mathbf{\Xi }_{k}^{-\frac{1}{2}}\mathbb{E} \left\{ \mathbf{G}_{kk} \right\} \mathbf{F}_{k,\mathrm{u}}$ to obtain an effective scalar channel as
\begin{equation}
\begin{aligned}
\left( \mathbf{\Xi }_{k}^{-\frac{1}{2}}\mathbb{E} \left\{ \mathbf{G}_{kk} \right\} \mathbf{F}_{k,\mathrm{u}} \right) ^H\mathbf{\Xi }_{k}^{-\frac{1}{2}}\mathbf{\hat{x}}_k=\left( \mathbb{E} \left\{ \mathbf{G}_{kk} \right\} \mathbf{F}_{k,\mathrm{u}} \right) ^H\mathbf{\Xi }_{k}^{-1}\mathbb{E} \left\{ \mathbf{G}_{kk} \right\} \mathbf{F}_{k,\mathrm{u}}\mathbf{x}_k+\left(\mathbb{E} \left\{ \mathbf{G}_{kk} \right\} \mathbf{F}_{k,\mathrm{u}} \right) ^H\mathbf{\Xi }_{k}^{-1}\mathbf{v}.
\end{aligned}
\end{equation}

Based on theories of optimal receivers \cite{tse2005fundamentals}, we derive optimal LSFD matrix $\mathbf{A}_k \!\!=\!\mathbf{\Xi }_{k}^{-1}\mathbb{E} \left\{ \mathbf{G}_{kk} \right\} \!\mathbf{F}_{k,\mathrm{u}}$ as
\begin{equation}\label{Optimal_LSFD_Proof1}
\begin{aligned}
\mathbf{A}_k=\left( \sum_{l=1}^K{\mathbb{E} \left\{ \mathbf{G}_{kl}\mathbf{F}_{k,\mathrm{u}}\mathbf{F}_{k,\mathrm{u}}^{H}\mathbf{G}_{kl}^{H} \right\}}-\mathbb{E} \left\{ \mathbf{G}_{kk} \right\} \mathbf{F}_{k,\mathrm{u}}\mathbf{F}_{k,\mathrm{u}}^{H}\mathbb{E} \left\{ \mathbf{G}_{kk}^{H} \right\} +\sigma ^2\mathbf{S}_k \right) ^{-1}\mathbb{E} \left\{ \mathbf{G}_{kk} \right\} \mathbf{F}_{k,\mathrm{u}}.
\end{aligned}
\end{equation}

Moreover, based on the the standard results of matrix derivation in \cite{hjorungnes2011complex}, we can easily obtain the LSFD matrix minimizing the conditional MSE for UE $k$ $\mathrm{MSE}_{k}^{(2)}=\mathrm{tr}( \mathbf{E}_{k}^{(2)})$ as
\begin{equation}\label{Optimal_LSFD_Proof2}
\begin{aligned}
\mathbf{A}_{k}=\left( \sum_{l=1}^K{\mathbb{E} \{ \mathbf{G}_{kl}\mathbf{\bar{F}}_{l,\mathrm{u}}\mathbf{G}_{kl}^{H} \}}+\sigma ^2\mathbf{S}_k \right) ^{-1}\mathbb{E} \{ \mathbf{G}_{kk} \} \mathbf{F}_{k,\mathrm{u}}.
\end{aligned}
\end{equation}

We notice that the LSFD matrix in \eqref{Optimal_LSFD_Proof1} is equivalent to the LSFD matrix in \eqref{Optimal_LSFD_Proof2}, except from having another scaling matrix $\mathbf{I}_N-\left( \mathbf{C}^H\mathbf{B}^{-1}\mathbf{C}+\mathbf{I}_{N} \right) ^{-1}\mathbf{C}^H\mathbf{B}^{-1}\mathbf{C}$ on the right side, which would not affect the value of \eqref{SE_LSFD_Origin}, where $\mathbf{B}=\sum_{l=1}^K{\mathbb{E} \{ \mathbf{G}_{kl}\mathbf{F}_{k,\mathrm{u}}\mathbf{F}_{k,\mathrm{u}}^{H}\mathbf{G}_{kl}^{H} \}}+\sigma ^2\mathbf{S}_k$ and $\mathbf{C}=\mathbb{E} \{ \mathbf{G}_{kk} \} \mathbf{F}_{k,\mathrm{u}}$. So the LSFD matrix in \eqref{Optimal_LSFD_Proof2} cannot maximize the achievable SE but minimize the MSE for UE $k$.

\section{Proof ot Theorem~\ref{Th_Closed_Form}}\label{Appendix2}
In this part, we compute terms of \eqref{SE_LSFD_Origin} in closed-form for the LSFD scheme with MR combining $\mathbf{V}_{mk}=\mathbf{\hat{H}}_{mk}$. For the first term $\mathbf{D}_{k,(2)}=\mathbf{A}_{k}^{H}\mathbb{E} \{ \mathbf{G}_{kk} \} \mathbf{F}_{k,\mathrm{u}}$, we have $\mathbb{E}\{ \mathbf{G}_{kk} \} =[ \mathbb{E}\{ \mathbf{V}_{1k}^{H}\mathbf{H}_{1k} \} ;\dots ;\mathbb{E}\{ \mathbf{V}_{Mk}^{H}\mathbf{H}_{Mk} \} ] =[ \mathbf{Z}^{T}_{1k},\dots ,\mathbf{Z}^{T}_{Mk} ]^{T}\triangleq \mathbf{Z}_k$, where $\mathbf{Z}_{mk}=\mathbb{E}\{ \mathbf{V}_{mk}^{H}\mathbf{H}_{mk} \} =\mathbb{E}\{ \mathbf{\hat{H}}_{mk}^{H}\mathbf{\hat{H}}_{mk} \} \in \mathbb{C}^{N\times N}$ and the $\left( n,n^{\prime} \right) $-th element of $\mathbf{Z}_{mk}$ can be denoted as
$[ \mathbf{Z}_{mk} ] _{nn^{\prime}}=\mathbb{E}\{ \mathbf{\hat{h}}_{mk,n}^{H}\mathbf{\hat{h}}_{mk,n^{\prime}}\} =\mathrm{tr}( \mathbf{\hat{R}}_{mk}^{n^{\prime}n} )$. So we derive the closed-form for $\mathbf{D}_{k,(2)}$ as
$\mathbf{D}_{k,(2),\mathrm{c}}=\mathbf{A}_{k}^{H}\mathbf{Z}_k\mathbf{F}_{k,\mathrm{u}}$. As for the second term $\mathbf{S}_k\in \mathbb{C}^{MN\times MN} $, we have
$
\mathbf{S}_k =\mathrm{diag}( \mathbb{E}\{ \mathbf{V}_{1k}^{H}\mathbf{V}_{1k} \} ,\dots ,\mathbb{E}\{ \mathbf{V}_{Mk}^{H}\mathbf{V}_{Mk}\} )=\mathrm{diag}( \mathbf{Z}_{1k},\dots ,\mathbf{Z}_{Mk}).
$
For $\mathbb{E} \{ \mathbf{G}_{kl}\mathbf{\bar{F}}_{l,\mathrm{u}}\mathbf{G}_{kl}^{H} \}$, we notice that the $( m,m^{\prime})$-submatrix of $\mathbb{E} \{ \mathbf{G}_{kl}\mathbf{\bar{F}}_{l,\mathrm{u}}\mathbf{G}_{kl}^{H} \}$ is $\mathbb{E}\{ \mathbf{V}_{mk}^{H}\mathbf{H}_{ml}\mathbf{\bar{F}}_{l,\mathrm{u}}\mathbf{H}_{m^{\prime}l}^{H}\mathbf{V}_{m^{\prime}k} \} $.

 Based on \cite{04962}, we compute $\mathbb{E}\{ \mathbf{V}_{mk}^{H}\mathbf{H}_{ml}\mathbf{\bar{F}}_{l,\mathrm{u}}\mathbf{H}_{m^{\prime}l}^{H}\mathbf{V}_{m^{\prime}k} \} $ for four possible AP-UE combinations. For ``$m\ne m^{\prime},l\notin \mathcal{P}_k$", we have $\mathbb{E}\{ \mathbf{V}_{mk}^{H}\mathbf{H}_{ml}\mathbf{\bar{F}}_{l,\mathrm{u}}\mathbf{H}_{m^{\prime}l}^{H}\mathbf{V}_{m^{\prime}k} \} =0$ for the independence between $\mathbf{V}_{mk}$ and $\mathbf{H}_{ml}$. For ``$m\ne m^{\prime},l\in \mathcal{P}_k$", we have $\mathbb{E}\{ \mathbf{V}_{mk}^{H}\mathbf{H}_{ml}\mathbf{\bar{F}}_{l,\mathrm{u}}\mathbf{H}_{m^{\prime}l}^{H}\mathbf{V}_{m^{\prime}k} \} =\mathbb{E}\{ \mathbf{V}_{mk}^{H}\mathbf{H}_{ml} \} \mathbf{\bar{F}}_{l,\mathrm{u}}\mathbb{E}\{ \mathbf{H}_{m^{\prime}l}^{H}\mathbf{V}_{m^{\prime}k}\} =\mathbf{\Lambda }_{mkl}\mathbf{\bar{F}}_{l,\mathrm{u}}\mathbf{\Lambda }_{m^{\prime}lk}$, where the $\left( n,n^{\prime} \right) $-th element of $N\times N$-dimension complex matrices $\mathbf{\Lambda }_{mkl}\triangleq \mathbb{E}\{ \mathbf{V}_{mk}^{H}\mathbf{H}_{ml} \}$, $\mathbf{\Lambda }_{m^{\prime}lk} \triangleq  \mathbb{E}\{ \mathbf{H}_{m^{\prime}l}^{H}\mathbf{V}_{m^{\prime}k}\}$ are $[\mathbf{\Lambda }_{mkl}]_{nn^{\prime}}=\mathbb{E}\{ \mathbf{\hat{h}}_{mk,n}^{H}\mathbf{\hat{h}}_{ml,n^{\prime}} \} =\mathrm{tr(}\mathbf{\Xi }_{mkl}^{n^{\prime}n})$ and $[\mathbf{\Lambda }_{m^{\prime}lk}]_{nn^{\prime}}=\mathbb{E} \{\mathbf{\hat{h}}_{m^{\prime}l,n}^{H}\mathbf{\hat{h}}_{mk,n^{\prime}}\}=\mathrm{tr(}\mathbf{\Xi }_{m^{\prime}lk}^{n^{\prime}n})$
with $\mathbf{\Xi }_{mkl}\triangleq \mathbb{E}\{ \mathbf{\hat{h}}_{ml}\mathbf{\hat{h}}_{mk}^{H} \}=\tau _p\mathbf{R}_{ml}\mathbf{\tilde{F}}_{l,\mathrm{p}}^{H}\mathbf{\Psi }_{mk}^{-1}\mathbf{\tilde{F}}_{k,\mathrm{p}}\mathbf{R}_{mk}$, $\mathbf{\Xi }_{m^{\prime}lk}\triangleq \mathbb{E} \{ \mathbf{\hat{h}}_{m^{\prime}k}\mathbf{\hat{h}}_{m^{\prime}l}^{H}\}=\tau _p\mathbf{R}_{m^{\prime}k}\mathbf{\tilde{F}}_{k,\mathrm{p}}^{H}\mathbf{\Psi }_{m^{\prime}k}^{-1}\mathbf{\tilde{F}}_{l,\mathrm{p}}\mathbf{R}_{m^{\prime}l}$.
For ``$m=m^{\prime},l\notin \mathcal{P}_k$", we define
$\mathbf{\Gamma }_{mkl}^{\left( 1 \right)}\triangleq \mathbb{E}\{ \mathbf{V}_{mk}^{H}\mathbf{H}_{ml}\mathbf{\bar{F}}_{l,\mathrm{u}}\mathbf{H}_{ml}^{H}\mathbf{V}_{mk}\} \in \mathbb{C}^{N\times N}$ with the $\left( n,n^{\prime}\right)$-th element $[\mathbf{\Gamma }_{mkl}^{\left( 1 \right)}]_{nn^{\prime}}=\sum_{i=1}^N{\sum_{i^{\prime}=1}^N{[\mathbf{\bar{F}}_{l,\mathrm{u}}]_{i^{\prime}i}\mathbb{E} \{\mathbf{\hat{h}}_{mk,n}^{H}\mathbf{h}_{ml,i^{\prime}}\mathbf{h}_{ml,i}^{H}\mathbf{\hat{h}}_{mk,n^{\prime}}\}}}$ being
\begin{equation}
\begin{aligned}
[\mathbf{\Gamma }_{mkl}^{\left( 1 \right)}]_{nn^{\prime}}=\sum_{i=1}^N{\sum_{i^{\prime}=1}^N{[\mathbf{\bar{F}}_{l,\mathrm{u}}]_{i^{\prime}i}\mathrm{tr}( \mathbb{E} \left\{ \mathbf{h}_{ml,i^{\prime}}\mathbf{h}_{ml,i}^{H} \right\} \mathbb{E} \{ \mathbf{\hat{h}}_{mk,n^{\prime}}\mathbf{\hat{h}}_{mk,n}^{H} \})}}=\sum_{i=1}^N{\sum_{i^{\prime}=1}^N{[\mathbf{\bar{F}}_{l,\mathrm{u}}]_{i^{\prime}i}\mathrm{tr}( \mathbf{R}_{ml}^{i^{\prime}i}\mathbf{\hat{R}}_{mk}^{n^{\prime}n} )}}
\end{aligned}
\end{equation}
since $\mathbf{\hat{H}}_{mk}$ and $\mathbf{H}_{ml}$ are independent. Finally, for ``$m=m^{\prime},l\in \mathcal{P}_k$", $\mathbf{\hat{H}}_{mk}$ and $\mathbf{H}_{ml}$ are no longer independent. We define $\mathbf{\Gamma }_{mkl}^{\left( 2 \right)}\triangleq \mathbb{E}\{ \mathbf{V}_{mk}^{H}\mathbf{H}_{ml}\mathbf{\bar{F}}_{l,\mathrm{u}}\mathbf{H}_{ml}^{H}\mathbf{V}_{mk} \} \in \mathbb{C}^{N\times N}$ whose $( n,n^{\prime})$-th element is
$
[ \mathbf{\Gamma }_{mkl}^{\left( 2 \right)}] _{nn^{\prime}}=\sum_{i=1}^N{\sum_{i^{\prime}=1}^N{[\mathbf{\bar{F}}_{l,\mathrm{u}}] _{i^{\prime}i}\mathbb{E} \{ \mathbf{\hat{h}}_{mk,n}^{H}\mathbf{h}_{ml,i^{\prime}}\mathbf{h}_{ml,i}^{H}\mathbf{\hat{h}}_{mk,n^{\prime}}\}}}.
$
We follow the similar method in \cite{04962} and derive\\
$
[\mathbf{\Gamma }_{kl,m}^{\left( 2 \right)}]_{nn^{\prime}}=\sum_{i=1}^N{\sum_{i^{\prime}=1}^N{[\mathbf{\bar{F}}_{l,\mathrm{u}}]_{i^{\prime}i}\mathrm{tr(}\mathbf{R}_{ml}^{i^{\prime}i}\mathbf{P}_{mkl,\left( 1 \right)}^{n^{\prime}n})}}+\tau _{p}^{2}\sum_{q_1=1}^N{\sum_{q_2=1}^N{[\mathbf{\bar{F}}_{l,\mathrm{u}}]_{i^{\prime}i}[\mathrm{tr(}\mathbf{\tilde{P}}_{mkl,\left( 2 \right)}^{q_1n}\mathbf{\tilde{R}}_{ml}^{i^{\prime}q_2}\mathbf{\tilde{R}}_{ml}^{q_2i}\mathbf{\tilde{P}}_{mkl,\left( 2 \right)}^{n^{\prime}q_1})]}}.
+\tau _{p}^{2}\sum_{q_1=1}^N{\sum_{q_2=1}^N{[\mathbf{\bar{F}}_{l,\mathrm{u}}]_{i^{\prime}i}\mathrm{tr(}\mathbf{\tilde{P}}_{mkl,\left( 2 \right)}^{q_1n}\mathbf{\tilde{R}}_{ml}^{i^{\prime}q_2})\mathrm{tr(}\mathbf{\tilde{P}}_{mkl,\left( 2 \right)}^{n^{\prime}q_2}\mathbf{\tilde{R}}_{ml}^{q_2i})}},
$
%
where $\mathbf{P}_{mkl,( 1 )}\triangleq\tau _p\mathbf{S}_{mk}( \mathbf{\Psi }_{mk}-\tau _p\mathbf{\tilde{F}}_{l,\mathrm{p}}\mathbf{R}_{ml}\mathbf{\tilde{F}}_{l,\mathrm{p}}^{H} ) \mathbf{S}_{mk}^{H}$, $\mathbf{S}_{mk}\triangleq\mathbf{R}_{mk}\mathbf{\tilde{F}}_{k,\mathrm{p}}^{H}\mathbf{\Psi }_{mk}^{-1}$ and $\mathbf{P}_{mkl,( 2 )}\triangleq\mathbf{S}_{mk}\mathbf{\tilde{F}}_{l,\mathrm{p}}\mathbf{R}_{ml}\mathbf{\tilde{F}}_{l,\mathrm{p}}^{H}\mathbf{S}_{mk}^{H}$, respectively. Besides, $\mathbf{\tilde{R}}_{ml}^{ni}$ and $\mathbf{\tilde{P}}_{mkl,( 2 )}^{ni}$ denote $( n,i )$-submatrix of $\mathbf{R}_{ml}^{\frac{1}{2}}$ and $\mathbf{P}_{mkl,( 2 )}^{\frac{1}{2}}$, respectively.

In summary, combining all the cases, we have $\mathbb{E} \{ \mathbf{G}_{kl}\mathbf{\bar{F}}_{l,\mathrm{u}}\mathbf{G}_{kl}^{H} \} =\mathbf{T}_{kl,( 1 )}+\mathbf{T}_{kl,( 2 )}$ if $l\in \mathcal{P} _k$ and $\mathbf{T}_{kl,( 1 )}$ otherwise, where $\mathbf{T}_{kl,\left( 1 \right)}\triangleq \mathrm{diag}( \mathbf{\Gamma }_{kl,1}^{( 1 )},\dots ,\mathbf{\Gamma }_{kl,M}^{( 1 )} ) \in \mathbb{C} ^{MN\times MN}$ and $\mathbf{T}_{kl,\left( 2 \right)}^{mm^{\prime}}=\mathbf{\Gamma }_{kl,m}^{\left( 2 \right)}-\mathbf{\Gamma }_{kl,m}^{\left( 1 \right)}$ if $m=m^{\prime}$ and $\mathbf{\Lambda }_{mkl}\mathbf{\bar{F}}_{l,\mathrm{u}}\mathbf{\Lambda }_{m^{\prime}lk}$ otherwise. Plugging the derived results into \eqref{Optimal_LSFD} and \eqref{MMSE_MSE_Matrix}, we can easily compute the optimal LSFD coefficient matrix and MSE matrix in closed-form as \eqref{Closed_form_LSFD_MSE}. So we have finished the proof of Theorem~\ref{Th_Closed_Form}. For more details on the derived expression, please refer to \cite[Appendix D]{04962}.
\vspace{-0.5cm}
\section{Proof of Theorem~\ref{Optimal_Precoding_1}}\label{Appendix3}
When other optimization variables are fixed, we derive the partial derivative of \eqref{Lagrange_Function_1} w.r.t $\mathbf{F}_{k,\mathrm{u}}^{(1)}$ as
\begin{equation}
\begin{aligned}
\frac{\partial f\left( \mathbf{F}_{1,\mathrm{u},(1)},\dots ,\mathbf{F}_{K,\mathrm{u},(1)} \right)}{\partial \mathbf{F}_{k,\mathrm{u},(1)}}&=\sum_{l=1}^K{\mu _{l,(1)}\left( \mathbf{\hat{H}}_{k}^{H}\mathbf{V}_l\mathbf{W}_{l,(1)}\mathbf{V}_{l}^{H}\mathbf{\hat{H}}_k+\mathbb{E} \left\{ \mathbf{\tilde{H}}_{k}^{H}\mathbf{V}_l\mathbf{W}_{l,(1)}\mathbf{V}_{l}^{H}\left. \mathbf{\tilde{H}}_k \right|\mathbf{V},\mathbf{W} \right\} \right)}+\lambda _{k,(1)}\mathbf{I}_N\\
&-\mu _{k,(1)}\mathbf{\hat{H}}_{k}^{H}\mathbf{V}_{k}^{H}\mathbf{W}_{k,(1)}.
\end{aligned}
\end{equation}
By applying the first-order optimality condition and setting $\frac{\partial f\left( \mathbf{F}_{1,\mathrm{u},(1)},\dots ,\mathbf{F}_{K,\mathrm{u},(1)} \right)}{\partial \mathbf{F}_{k,\mathrm{u},(1)}}=0$, we can easily obtain the optimal precoding scheme. Besides, $\lambda _{k,(1)}$ and $\mathbf{F}_{k,\mathrm{u},(1)}$ should also satisfy KKT condition as \eqref{KKT_1}.

As for $\mathbf{\bar{C}}_{kl}\triangleq \mathbb{E} \{ \mathbf{\tilde{H}}_{k}^{H}\mathbf{V}_l\mathbf{W}_{l,(1)}\mathbf{V}_{l}^{H}\mathbf{\tilde{H}}_k |\mathbf{V},\mathbf{W} \} \in \mathbb{C} ^{N\times N}$, by applying Lemma~\ref{Lemma1}, the $(i,n)$-th element of $\mathbf{\bar{C}}_{kl}$ is $\mathrm{tr}( \mathbf{\bar{V}}_l\mathbb{E} \{ \mathbf{\tilde{h}}_{k,n}\mathbf{\tilde{h}}_{k,i}^{H}\})$ where $\mathbf{\bar{V}}_l\triangleq \mathbf{V}_l\mathbf{W}_{l,(1)}\mathbf{V}_{l}^{H}$ and $\mathbf{\tilde{h}}_{k,n}=[ \mathbf{\tilde{h}}_{1k,n}^{T},\dots ,\mathbf{\tilde{h}}_{Mk,n}^{T}] ^T\in \mathbb{C} ^{ML}$ is the $n$-th column of $\mathbf{\tilde{H}}_{k}$. Finally, we derive $\mathbf{C}_{k,ni}\triangleq \mathbb{E} \{ \mathbf{\tilde{h}}_{k,n}\mathbf{\tilde{h}}_{k,i}^{H}\} =\mathrm{diag}\left( \mathbf{C}_{1k}^{ni},\dots ,\mathbf{C}_{Mk}^{ni} \right)\in \mathbb{C} ^{ML\times ML}$ since $\mathbf{\tilde{h}}_{mk,n}$ and $\mathbf{\tilde{h}}_{m^{\prime}k,n}$ for $m\ne m^{\prime}$ are independent and both have zero mean. So $\mathbf{C}_{k,ni}$ is a block-diagonal matrix with the square matrices $\mathbf{C}_{1k}^{ni}=\mathbb{E} \{ \mathbf{\tilde{h}}_{1k,n}\mathbf{\tilde{h}}_{1k,i}^{H} \} ,\dots ,\mathbf{C}_{Mk}^{ni}=\mathbb{E} \{ \mathbf{\tilde{h}}_{Mk,n}\mathbf{\tilde{h}}_{Mk,i}^{H}\}$ on the diagonal.

\section{Proof of \eqref{SE_MSE}}\label{MSE_SE}
For the LSFD scheme, the conditional MSE matrix for UE $k$ can be written as \eqref{MSE_Matrix}. Based on \cite[Appendix C]{04962}, we prove that \eqref{Optimal_LSFD} can also minimize $\mathrm{MSE}_{k,(2)}=\mathrm{tr}\left( \mathbf{E}_{k,(2)} \right)$. With \eqref{Optimal_LSFD} implemented, $\mathbf{E}_{k,(2)} $ is given by \eqref{MMSE_MSE_Matrix}. Then, by applying Lemma~\ref{Lemma2}, we have
\begin{equation}
\begin{aligned}\label{E2_Inverse}
\left( \mathbf{E}_{k,(2)}^{\mathrm{opt}} \right) ^{-1}=\mathbf{I}_N+\mathbf{F}_{k,\mathrm{u},(2)}^{H}\mathbb{E} \left\{ \mathbf{G}_{kk}^{H} \right\} &\left( \sum_{l=1}^K{\mathbb{E} \left\{ \mathbf{G}_{kl}\mathbf{\bar{F}}_{l,\mathrm{u},(2)}\mathbf{G}_{kl}^{H} \right\}}-\mathbb{E} \left\{ \mathbf{G}_{kk} \right\} \mathbf{\bar{F}}_{k,\mathrm{u},(2)}\mathbb{E} \left\{ \mathbf{G}_{kk}^{H} \right\} +\sigma ^2\mathbf{S}_k \right) ^{-1}\notag \\
&\times \mathbb{E} \left\{ \mathbf{G}_{kk} \right\} \mathbf{F}_{k,\mathrm{u},(2)},
\end{aligned}
\end{equation}
where $\mathbf{A} \triangleq \mathbf{I}_N$, $\mathbf{B} \triangleq -\mathbf{F}_{k,\mathrm{u},(2)}^{H}\mathbb{E} \{ \mathbf{G}_{kk}^{H} \}$, $\mathbf{C} \triangleq ( \sum_{l=1}^K{\mathbb{E} \{ \mathbf{G}_{kl}\mathbf{\bar{F}}_{l,\mathrm{u},(2)}\mathbf{G}_{kl}^{H} \}}+\sigma ^2\mathbf{S}_k ) ^{-1}$ and $\mathbf{D} \triangleq \mathbb{E}\{ \mathbf{G}_{kk} \} \mathbf{F}_{k,\mathrm{u},(2)}$, respectively. We show the equivalence between $\mathrm{SE}_{k,(2)}^{\mathrm{opt}}$ and $\log _2| ( \mathbf{E}_{k,(2)}^{\mathrm{opt}} ) ^{-1} |$ without a factor $( 1-{\tau _p}/{\tau _c} )$.
\vspace{-0.3cm}
\section{Proof of Theorem~\ref{F_Th_Closed_Form}}\label{F_Closed_Form}
When MR combining $\mathbf{V}_{mk}=\mathbf{\hat{H}}_{mk}$ and the optimal LSFD scheme applied, we can easily compute $\mathbb{E} \{ \mathbf{G}_{kk}^{H} \}$, $\mathbf{A}_k^{\mathrm{opt}}$, and $\mathbf{E}_{k,(2)}^{\mathrm{opt}}$ in closed-form as Theorem~\ref{Th_Closed_Form}. Furthermore, by applying Lemma~\ref{Lemma1}, the $(i,n)$-th entry of $\mathbf{\bar{T}}_{lk}=\mathbb{E} \{\mathbf{G}_{lk}^{H}\mathbf{A}_l\mathbf{E}_{l,(2)}^{-1}\mathbf{A}_{l}^{H}\mathbf{G}_{lk}\}\in\mathbb{C} ^{N\times N}$ can be denoted as $\mathrm{tr}( \mathbf{\bar{A}}_l\mathbb{E} \{ \mathbf{g}_{lk,n}\mathbf{g}_{lk,i}^{H} \})$, where $\mathbf{\bar{A}}_l\triangleq \mathbf{A}_l\mathbf{E}_{l,(2)}^{-1}\mathbf{A}_{l}^{H}$ and $\mathbf{g}_{lk,n}\in \mathbb{C} ^{MN}$ is the $n$-th column of $\mathbf{G}_{lk}$. Note that the $\left( m-1 \right) N+p$-th element of $\mathbf{g}_{lk,n}$ is $\mathbf{\hat{h}}_{ml,p}^{H}\mathbf{h}_{mk,n}$ so the $[ \left( m-1 \right) N+p,\left( m^{\prime}-1 \right) N+p^{\prime}] $-th (or $[o,j]$-th briefly) entry of $\mathbf{\bar{G}}_{lk,ni}\triangleq \mathbb{E} \{ \mathbf{g}_{lk,n}\mathbf{g}_{lk,i}^{H} \} \in \mathbb{C} ^{MN\times MN}$ can be denoted as $\mathbb{E} \{ \mathbf{\hat{h}}_{ml,p}^{H}\mathbf{h}_{mk,n}\mathbf{h}_{m^{\prime}k,i}^{H}\mathbf{\hat{h}}_{m^{\prime}l,p^{\prime}} \}$, which can be computed for four AP-UE combinations as Theorem~\ref{Th_Closed_Form}.

For ``$l\notin \mathcal{P}_k,m\ne m^{\prime}$", we have $\mathbb{E}\{ \mathbf{\hat{h}}_{ml,p}^{H}\mathbf{h}_{mk,n}\mathbf{h}_{m^{\prime}k,i}^{H}\mathbf{\hat{h}}_{m^{\prime}l,p^{\prime}}\} =0$. For ``$l\in \mathcal{P}_k,m\ne m^{\prime}$", we have $\mathbb{E}\{ \mathbf{\hat{h}}_{ml,p}^{H}\mathbf{h}_{mk,n}\mathbf{h}_{mk,i}^{H}\mathbf{\hat{h}}_{ml,p^{\prime}}\} =\mathrm{tr}( \mathbf{R}_{mk}^{ni}\mathbf{\hat{R}}_{ml}^{p^{\prime}p})$. For ``$l\notin \mathcal{P}_k,m=m^{\prime}$", we have
$\mathbb{E} \{ \mathbf{\hat{h}}_{ml,p}^{H}\mathbf{h}_{mk,n}\mathbf{h}_{m^{\prime}k,i}^{H}\mathbf{\hat{h}}_{m^{\prime}l,p^{\prime}}\}=\mathbb{E}\{ \mathbf{\hat{h}}_{ml,p}^{H}\mathbf{h}_{mk,n} \} \mathbb{E} \{ \mathbf{h}_{m^{\prime}k,i}^{H}\mathbf{\hat{h}}_{m^{\prime}l,p^{\prime}}\}=\mathrm{tr}( \mathbf{\Xi }_{mlk}^{np}) \mathrm{tr}( \mathbf{\Xi }_{m^{\prime}kl}^{p^{\prime}i} )$, where $\mathbf{\Xi }_{mlk}=\tau _p\mathbf{R}_{mk}\mathbf{\tilde{F}}_{k,\mathrm{p}}^{H}\mathbf{\Psi }_{mk}^{-1}\mathbf{\tilde{F}}_{l,\mathrm{p}}\mathbf{R}_{ml}$ and $\mathbf{\Xi }_{m^{\prime}kl}=\tau _p\mathbf{R}_{m^{\prime}l}\mathbf{\tilde{F}}_{l,\mathrm{p}}^{H}\mathbf{\Psi }_{m^{\prime}l}^{-1}\mathbf{\tilde{F}}_{k,\mathrm{p}}\mathbf{R}_{m^{\prime}k}$. For ``$l\in \mathcal{P}_k,m=m^{\prime}$", we obtain $\mathbb{E} \{ \mathbf{\hat{h}}_{ml,p}^{H}\mathbf{h}_{mk,n}\mathbf{h}_{mk,i}^{H}\mathbf{\hat{h}}_{ml,p^{\prime}}\}=\mathrm{tr}( \mathbf{R}_{mk}^{ni}\mathbf{P}_{mkl,\left( 1 \right)}^{p^{\prime}p} ) +\tau _{p}^{2}\sum_{q_1=1}^N{\sum_{q_2=1}^N{\mathrm{tr}( \mathbf{\tilde{P}}_{mlk,\left( 2 \right)}^{q_1p}\mathbf{\tilde{R}}_{mk}^{nq_2}\mathbf{\tilde{R}}_{mk}^{q_2i}\mathbf{\tilde{P}}_{mlk,\left( 2 \right)}^{p^{\prime}q_1} )}}+\tau _{p}^{2}\sum_{q_1=1}^N{\sum_{q_2=1}^N{\mathrm{tr}( \mathbf{\tilde{P}}_{mlk,\left( 2 \right)}^{q_1n}\mathbf{\tilde{R}}_{mk}^{nq_1}) \mathrm{tr}( \mathbf{\tilde{P}}_{mlk,\left( 2 \right)}^{p^{\prime}q_2}\mathbf{\tilde{R}}_{mk}^{q_2i})}},$
where $\mathbf{S}_{ml}=\mathbf{R}_{ml}\mathbf{\tilde{F}}_{l,\mathrm{p}}^{H}\mathbf{\Psi }_{ml}^{-1}$, $\mathbf{P}_{mlk,\left( 1 \right)}=\tau _p\mathbf{S}_{ml}( \mathbf{\Psi }_{ml}-\tau _p\mathbf{\tilde{F}}_{k,\mathrm{p}}\mathbf{R}_{mk}\mathbf{\tilde{F}}_{k,\mathrm{p}}^{H} ) \mathbf{S}_{ml}^{H}$ and $\mathbf{P}_{mlk,\left( 2 \right)}=\mathbf{S}_{ml}\mathbf{\tilde{F}}_{k,\mathrm{p}}\mathbf{R}_{mk}\mathbf{\tilde{F}}_{k,\mathrm{p}}^{H}\mathbf{S}_{ml}^{H}$ with
$\mathbf{\tilde{R}}_{mk}^{ni}$ and $\mathbf{\tilde{P}}_{mkl,( 2 )}^{ni}$ being $( n,i )$-submatrix of $\mathbf{R}_{mk}^{\frac{1}{2}}$ and $\mathbf{P}_{mkl,( 2 )}^{\frac{1}{2}}$, respectively. We can compute $\mathbb{E} \{\mathbf{g}_{lk,n}\mathbf{g}_{lk,i}^{H}\}_{oj}$ in closed-form as \eqref{gg} and  $\mathbf{F}_{k,\mathrm{u},(2)}^{\mathrm{opt}}$ in closed-form as \eqref{F_2_Closed_form}.

\end{appendices}





\bibliographystyle{IEEEtran}
\bibliography{IEEEabrv,Ref}

\end{document}